\documentclass[12pt,a4paper]{article}

\usepackage{flonat-wp}
\usepackage[harvard]{flonat-bib}

\usetikzlibrary{positioning,arrows.meta,calc,fit,backgrounds}

\newcommand{\eps}{\varepsilon}
\newcommand{\Adv}{\mathcal{A}}
\newcommand{\Provider}{\mathcal{P}}
\newcommand{\Auditor}{\mathcal{V}}
\newcommand{\Index}{D}
\newcommand{\colluders}{k}

\newtheorem{observation}[theorem]{Observation}
\crefname{observation}{Observation}{Observations}
\Crefname{observation}{Observation}{Observations}

\title{Auditing Privacy in Multi-Tenant RAG under Account Collusion}

\author{
    Florian Burnat\thanks{University of Bath, Bath, UK. \texttt{fadb20@bath.ac.uk}.}
}

\date{31 May 2026}

\begin{document}

\begin{titlepage}
    \maketitle
    \begin{abstract}
        \noindent
Multi-tenant RAG services often treat the account as the privacy boundary: each account receives an $(\eps_{\text{acc}},\delta_{\text{acc}})$-DP retrieval guarantee against the tenant index. We show that this framing understates leakage under same-index account collusion. For Gaussian noise-then-select retrieval, $k$ coordinated same-tenant accounts compose to joint leakage $\Theta(\sqrt{k}\,\eps_{\text{acc}})$, not $\eps_{\text{acc}}$; we give a matching membership-inference attack and validate the predicted $\sqrt{k}$ AUC trend in scalar, top-$K$, trained-embedder, and production-scale HNSW settings. We then give a verifier-runnable audit protocol that attests noise-then-select retrieval and reports $(\textsf{PASS},\eps_{\text{audit}})$ for coalitions up to a declared cap $k_{\max}$, without disclosing the index or changing the retrieval decision rule. The claim is retrieval-channel only: generation-channel leakage and adversarially robust coalition-size estimation are complementary audit predicates.
    \end{abstract}
    \setcounter{page}{0}
    \thispagestyle{empty}
\end{titlepage}

\clearpage
\doublespacing

\section{Introduction}
\label{sec:introduction}

Production retrieval-augmented generation services, including Microsoft 365 Copilot, OpenAI Assistants, Anthropic Files API, and the managed-RAG products built atop Pinecone, Weaviate, and similar vector databases, are overwhelmingly multi-tenant. A single deployed embedder, a single deployed generator, and a fleet of per-tenant indices serve thousands to millions of customer accounts under a shared infrastructure. The DP-RAG designs proposed and audited at this architecture reason at per-account granularity: each account is rate-limited, each query is scored against only that account's tenant index, and retrieval scores are noised to deliver an $(\eps_{\text{acc}}, \delta_{\text{acc}})$-differential-privacy guarantee per account over an audit window \parencite{Cheng2025-remoterag,Zeng2024-good}. Our concern is the adequacy of the per-account framing where it applies, not that every named service formally advertises such a guarantee.

This per-account framing is computationally convenient and aligns with how providers reason about API contracts. It is also the framing implicit in the recent wave of single-attacker membership-inference work \parencite{Li2025-generating,Wang2025-leaks,Feng2025-leak,Liu2024-mask,Naseh2025-riddle,Gao2025-dcmi}, in which one isolated attacker is treated as the unit of analysis, competing on AUC at a fixed query budget under a $\colluders=1$ threat model. Existing defenses \parencite{Cheng2025-remoterag} likewise modify the deployment to leak less per query; however, none formalize the changes when the assumption of a single attacker is relaxed.

\paragraph{The gap} Production RAG accounts are cheap. An adversary with a budget for $\colluders$ disposable accounts, whether a malicious customer creating sock puppets, a malicious bystander tenant probing across the boundary, or an external attacker buying API access through resellers, can pool $\colluders$ accounts' query budgets against a single target tenant. The provider's per-account DP guarantee provides no direct information about what the coalition observes jointly. Worse, the correct benchmark is unclear: a basic group-privacy bound provides $\eps_{\text{joint}} \leq \colluders \cdot \eps_{\text{acc}}$, but this is loose; a naive ``per-account independence'' reading provides $\eps_{\text{joint}} \approx \eps_{\text{acc}}$, but this is incorrect. Neither extreme is the right benchmark: our results show the correct rate is $\Theta(\sqrt{\colluders}\cdot\eps_{\text{acc}})$. Concretely, a service advertising $\eps_{\text{acc}} = 1$ per account can in fact leak as if $\eps \approx 3.16$ to a ten-account coalition at realistic $\delta_{\text{acc}} = 10^{-6}$, even though no single account exceeds its budget; for a fifty-account coalition the joint $\eps$ rises to ${\approx}7.07$ at the same per-account budget --- the privacy debt the provider implicitly takes on by failing to bound coalition size.

\paragraph{Contributions}
\begin{enumerate}[leftmargin=1.5em,label=\textbf{C\arabic*.}]
    \item \emph{Threat model: same-index multi-account collusion (\Cref{sec:threat-model,sec:dp-analysis}).} We formalize same-index multi-account collusion for multi-tenant RAG, define $(\eps, \delta, \colluders)$-collusion DP, and show that per-account DP structurally understates coalition leakage: for Gaussian noise-then-select retrieval, joint leakage scales as $\Theta(\sqrt{\colluders}\cdot\eps_{\text{acc}})$ (\Cref{thm:upper,thm:lower,cor:tightness}), unconditional for same-tenant collusion and under M4 (\Cref{rem:access-control}) otherwise. The composition rate specializes prior DP theory~\parencite{DworkRothblumVadhan2010,KairouzOhViswanath2017,VadhanWang2021-concurrent}; the contribution is the threat model and the bridge to RAG-specific MIA.
    \item \emph{Falsifiable empirical predictions (\Cref{sec:dp-analysis,sec:attacks}).} The DP-to-MIA reduction (\Cref{cor:mia-upper}) predicts membership-inference AUC of $\tfrac{1}{2} + \Theta(\sqrt{\colluders} \cdot \eps_{\text{acc}})$ at fixed per-account budget. We commit to two falsification gates (P1, P3, \Cref{subsec:mia-bridge}) and validate the $\sqrt{\colluders}$ trend across scalar, top-$K$, trained-embedder, and production-scale HNSW settings.
    \item \emph{Verifier-runnable audit protocol (\Cref{sec:audit-protocol}).} We give a four-phase protocol that attests noise-then-select retrieval and issues a $(\textsf{PASS}, \eps_{\text{audit}})$ verdict without changing the retrieval decision rule or disclosing index, model, or pipeline. \Cref{thm:audit-soundness} shows that, if the true coalition size is at most the declared cap $\colluders_{\max}$, the verdict implies $(\eps_{\text{audit}}, \delta_{\text{policy}})$-collusion-DP. The coalition-size estimator is an operational gate: \Cref{lem:estimator-guarantee} calibrates it on P-A/B/C coalition patterns, while adversarially diversified coalitions remain a stated limitation. We also identify and correct a subtle $\delta$-accounting error in naive multi-account composition.
\end{enumerate}

\section{Threat Model}
\label{sec:threat-model}

We formalize multi-tenant retrieval-augmented generation (RAG) under coordinated account collusion as a Stackelberg-like game between a service provider $\Provider$ and an adversary $\Adv$ controlling $\colluders$ accounts. The provider commits to a privacy policy, namely, per-account rate limits and a differentially private retrieval scoring mechanism, and the adversary best-responds with a coordinated query strategy. Our analysis (\Cref{sec:dp-analysis}) shows that the provider's per-account composition bound understates the total leakage, and the audit protocol (\Cref{sec:audit-protocol}) closes the resulting gap.

\subsection{System model}
\label{subsec:system-model}

A multi-tenant RAG service is a tuple $\Provider = (\Index, \mathsf{Emb}, \mathsf{LLM}, \mathsf{Policy})$ where:
\begin{itemize}[leftmargin=*]
    \item $\Index = \{\Index_t\}_{t \in T}$ is a family of per-tenant indices, $\Index_t \subseteq \mathcal{X}$, partitioned by tenant identifier $t \in T$.
    \item $\mathsf{Emb}: \mathcal{X} \to \mathbb{R}^d$ is a shared embedder.
    \item $\mathsf{LLM}: \mathcal{Q} \times \mathcal{X}^K \to \mathcal{Y}$ is a shared generator that conditions on a query $q$ and the top-$K$ retrieved documents.
    \item $\mathsf{Policy} = (r, \eps_{\text{acc}}, \delta_{\text{acc}}, \sigma)$ specifies the per-account rate limit $r$ (queries per audit window), per-account privacy budget $(\eps_{\text{acc}}, \delta_{\text{acc}})$, and noise parameter $\sigma$ on retrieval scores.
\end{itemize}

\Cref{fig:system-model} sketches the data flow. Accounts (rounded squares) live under tenants (dashed boundary boxes) and issue queries through a shared service; the service embeds each query, scores it against the requesting tenant's index, perturbs the scores with Gaussian noise \emph{before} top-$K$ selection, and feeds the selected documents to the LLM. A $\colluders$-coalition (red dashed link) is a subset of accounts that pool their responses through a shared output channel after the LLM has returned them, which the service cannot observe and the per-account rate limiter does not constrain.

\begin{figure}[t]
    \centering
    \begin{tikzpicture}[
        font=\footnotesize,
        every node/.style={inner sep=2pt},
        acc/.style={rounded corners=2pt, draw, minimum width=4mm, minimum height=4mm, fill=white},
        box/.style={draw, rounded corners=2pt, minimum height=6mm, minimum width=10mm, align=center, fill=white},
        bigbox/.style={draw, rounded corners=3pt, minimum height=14mm, minimum width=18mm, align=center, fill=blue!4},
        tenant/.style={draw, dashed, rounded corners=3pt, inner sep=4pt, fill=gray!8},
        flow/.style={-{Stealth[length=1.6mm]}, thick},
        collude/.style={red, dashed, thick},
    ]
        % Accounts grouped by tenant on the left
        \node[acc] (a1) at (0, 1.2) {};
        \node[acc, right=1mm of a1] (a2) {};
        \node[font=\scriptsize\itshape, above=0.5mm of $(a1)!0.5!(a2)$] {tenant $t_1$};

        \node[acc] (a3) at (0, 0) {};
        \node[acc, right=1mm of a3] (a4) {};
        \node[acc, right=1mm of a4] (a5) {};
        \node[font=\scriptsize\itshape, below=0.5mm of a4] {tenant $t_2$};

        \begin{scope}[on background layer]
            \node[tenant, fit=(a1)(a2), label={[font=\scriptsize\itshape]above:}] {};
            \node[tenant, fit=(a3)(a4)(a5)] {};
        \end{scope}

        % Service column
        \node[bigbox, right=12mm of a4] (svc) {RAG service};
        \node[box, below=1mm of svc.south, anchor=north, minimum width=18mm, fill=blue!4]
            (svcinner) {$\mathsf{Emb} \rightarrow \mathsf{Score} \rightarrow +\mathcal{N}(0,\sigma^2) \rightarrow \mathsf{TopK}$};

        % LLM
        \node[box, right=10mm of svc, minimum width=11mm, minimum height=10mm] (llm) {LLM};

        % Output channel
        \node[font=\scriptsize, right=8mm of llm, align=center] (out) {output\\channel};

        % Flows: accounts -> service
        \draw[flow] (a2.east) -- ($(svc.west)+(0,2mm)$);
        \draw[flow] (a5.east) -- ($(svc.west)+(0,-2mm)$);

        % service -> LLM -> output
        \draw[flow] (svc.east) -- (llm.west)
            node[midway, above, font=\scriptsize] {top-$K$};
        \draw[flow] (llm.east) -- (out.west);

        % Output back to accounts (per-account)
        \draw[flow, gray] (out.south) .. controls +(0,-6mm) and +(0,-6mm) ..
            ($(a4.south)+(0,-1mm)$);

        % Coalition link
        \draw[collude] (a1) .. controls +(0,-3mm) and +(0,3mm) .. (a5)
            node[midway, font=\scriptsize, red, fill=white, inner sep=0.5pt] {coalition $\colluders$};

        % Rate limiter annotation
        \node[font=\scriptsize, gray, above=0.5mm of svc.north] {per-account rate limit $r$};
    \end{tikzpicture}
    \caption{Multi-tenant RAG with a $\colluders$-account coalition. The service applies the rate limit per account and adds DP noise before top-$K$; coalitions pool outputs across per-account rate limits and privacy budgets after retrieval (red dashed link).}
    \label{fig:system-model}
\end{figure}
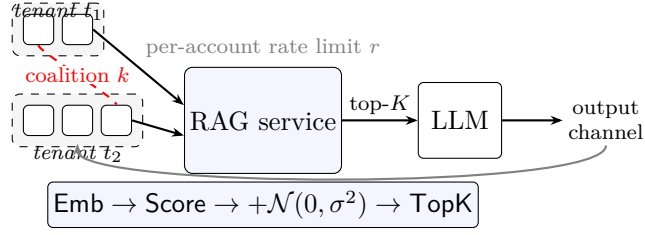

For a query $q$ from account $a$ with tenant $t(a)$, the service runs the retrieval mechanism
\begin{equation}
    \mathsf{Retr}(q, a) = \mathsf{TopK}\bigl(\mathsf{Score}(q, \Index_{t(a)}) + \mathcal{N}(0, \sigma^2 I)\bigr),
    \label{eq:retrieval}
\end{equation}
adding Gaussian noise \emph{before} top-$K$ selection,\footnote{Top-$K$ on noisy scores is a post-processing step over the DP output, so the privacy loss is non-increasing. Selecting first and noising only the $K$ winners breaks this --- the unnoised ordering leaks through the selected indices. The audit protocol attests the correct ordering.} and then returns $\mathsf{RAG}(q, a) = \mathsf{LLM}(q, \mathsf{Retr}(q, a))$. Here $\mathcal{N}(0, \sigma^2 I)$ denotes i.i.d.\ Gaussian noise on each coordinate of the $|\Index_{t(a)}|$-dimensional score vector; the scalar lower-bound mechanism $\mathcal{M}^*$ of \Cref{thm:lower} is the projection onto the single coordinate corresponding to the differing document. The per-account $(\eps_{\text{acc}}, \delta_{\text{acc}})$ guarantee and \Cref{thm:upper} govern $\mathsf{Retr}$; generation is out of scope (\Cref{subsec:scope-of-dp}).

\subsection{Provider strategy}
\label{subsec:provider-strategy}

The provider commits, before the audit window opens, to a policy $\mathsf{Policy}$ together with public commitments
\begin{equation}
    C_{\text{emb}} = H(\mathsf{Emb} \,\Vert\, \mathsf{prompt\_tmpl} \,\Vert\, \mathsf{tokenizer} \,\Vert\, \text{schema}(\Index)),
\end{equation}
\begin{equation}
    C_{\text{ledger},0} = H(\text{empty})
\end{equation}
where $C_{\text{ledger}}$ is the root of an append-only Merkle ledger of per-account query records (\Cref{sec:audit-protocol}). The provider's claimed privacy guarantee is that, for any single account $a$ over the audit window, the joint distribution of $a$'s \emph{retrieval-channel} transcripts ($\mathsf{Retr}$ outputs of \eqref{eq:retrieval}) satisfies $(\eps_{\text{acc}}, \delta_{\text{acc}})$-DP w.r.t.\ neighboring index pairs $\Index_{t(a)} \sim \Index'_{t(a)}$ differing by the addition, removal, or replacement of exactly one document (the standard add/remove/replace relation, $|\Index_{t(a)} \triangle \Index'_{t(a)}| \leq 2$); the per-account L2 sensitivity of $\mathsf{Score}(q, \cdot)$ over this relation is bounded by $1$ for unit-norm embeddings. Generation-channel privacy is out of scope (\Cref{subsec:scope-of-dp}).

\subsection{Adversary model}
\label{subsec:adversary-model}

The adversary $\Adv$ controls $\colluders \geq 1$ accounts $\{a_1, \ldots, a_\colluders\}$, each registered to a tenant $t(a_i)$ chosen by $\Adv$. We consider three regimes:
\begin{description}[leftmargin=1.5em,style=nextline]
    \item[Same-tenant collusion ($t(a_i) = t^* \forall i$)] $\Adv$ coordinates $\colluders$ accounts within a single tenant (e.g., internal users of an organization sharing an index $\Index_{t^*}$ of internally contributed documents) to extract documents from that tenant's index. Per-account DP is the contractual boundary the coalition crosses; the victim is other users of the same tenant.
    \item[Cross-tenant collusion ($t(a_i)$ varies)] $\Adv$ coordinates across tenants to extract another tenant's documents via shared infrastructure. Models a malicious bystander tenant.
    \item[External collusion ($t(a_i) \notin T_{\text{victim}}$)] $\Adv$ has no legitimate access to the victim tenant's index but uses $\colluders$ unrelated accounts. Models an external attacker.
\end{description}

\Cref{tab:regimes} summarises how each regime relates to the M4 access-control assumption (\Cref{rem:access-control}) and the resulting leakage about the victim index $\Index_{t^*}$.

\begin{table}[h]
    \centering
    \footnotesize
    \caption{Collusion regimes vs.\ M4 assumption. The DP analysis of \Cref{sec:dp-analysis} applies unconditionally to same-tenant collusion; cross-tenant and external regimes are vacuous about $\Index_{t^*}$ without M4 and identical to same-tenant under M4.}
    \label{tab:regimes}
    \begin{tabular}{l l c l}
    \toprule
    Regime & Tenant relation & M4? & Leakage about $\Index_{t^*}$ \\
    \midrule
    Same-tenant   & all $t(a_i){=}t^*$        & No  & $\Theta(\sqrt{\colluders}\,\eps_{\text{acc}})$ \\
    Cross-tenant  & some $t(a_i){\neq}t^*$    & Yes & same under M4; else $0$ \\
    External      & no $t(a_i){=}t^*$         & Yes & same under M4; else $0$ \\
    \bottomrule
    \end{tabular}
\end{table}

The adversary's capabilities are:
\begin{enumerate}[leftmargin=*,label=(C\arabic*)]
    \item \label{cap:posterior} \emph{Shared posterior}: all $\colluders$ accounts share a single probabilistic belief over the target index contents.
    \item \label{cap:adaptive} \emph{Adaptive coordinated queries}: the query of account $a_i$ at time $t$ may depend on the full transcript of all accounts up to time $t-1$.
    \item \label{cap:rate-aware} \emph{Rate-limit-aware scheduling}: $\Adv$ knows $r$ and schedules queries to remain within per-account limits while saturating aggregate budget.
    \item \label{cap:output} \emph{Shared output channel}: all responses across accounts are pooled before $\Adv$'s estimator runs.
\end{enumerate}

\paragraph{Modeling assumptions} The analysis of \Cref{sec:dp-analysis} treats the following as standing assumptions on the deployment: (\textbf{M1}) \emph{Independent fresh randomness}: each query consumes an independent draw of Gaussian noise; in particular, no two accounts (colluding or otherwise) share noise samples. (\textbf{M2}) \emph{Audit-window disjointness}: each audit window $W$ has its own per-account budget; the per-account $(\eps_{\text{acc}}, \delta_{\text{acc}})$ resets at window open, and adversary coordination is scoped to a single window unless explicitly noted. (\textbf{M3}) \emph{Embedder determinism within a window}: the embedder $\mathsf{Emb}$ and its parameters are fixed for the duration of $W$ and identical across tenants. (\textbf{M5}) \emph{Honest-account receipt sample}: at least one non-colluding account submits its query-time receipts to $\Auditor$; ledger completeness (\Cref{subsec:audit-attest}) is verified against this sample. These assumptions are standard for the DP-RAG setting and are enforced cryptographically by the audit protocol of \Cref{sec:audit-protocol} via $C_{\text{emb}}, C_{\text{policy}}$, and R2 receipts.

\begin{remark}[M4: same-index access for cross-tenant/external regimes]
\label{rem:access-control}
\Cref{eq:retrieval} confines an account's queries to $\Index_{t(a)}$, so under \emph{normal} access control cross-tenant and external accounts have zero leakage about the victim index $\Index_{t^*}$. We retain those regimes only under an explicit access-control-failure assumption --- (\textbf{M4}) \emph{Same-index access}: the attacker's accounts share retrieval access to $\Index_{t^*}$ via a shared backend (e.g., a misconfigured router or an over-broad shared index). Under M4, all three regimes are formally identical for \Cref{thm:upper}; the formal claims and audit protocol apply unconditionally only to same-tenant collusion. \Cref{subsec:external-collusion} reports the M4-failure experiment as an upper bound under access-control failure, not a guarantee about ordinary external attackers.
\end{remark}

\subsection{Adversary goals}
\label{subsec:adversary-goals}

We focus on membership inference (MIA): given a candidate document $x^*$, decide whether $x^* \in \Index_t$ for a target tenant $t$ (measured as ROC AUC or queries-to-confidence at a fixed FPR). More severe goals --- document reconstruction, cross-tenant leakage --- compose with our MIA bound but are formally analyzed only under M4 (\Cref{rem:access-control}).

\subsection{Game formulation}
\label{subsec:game}

The interaction unfolds over an audit window $W$ in the four phases specified in \Cref{sec:audit-protocol}: commitment, per-query attestation, coalition-size estimation, and verification.

\begin{definition}[$(\eps, \delta, \colluders)$-collusion DP]
\label{def:k-collusion-dp}
A multi-tenant RAG mechanism $\mathsf{RAG}$ is $(\eps, \delta, \colluders)$-collusion-DP with respect to tenant $t^*$ if, for any $\colluders$ adaptively-coordinated accounts $\{a_1, \ldots, a_\colluders\}$ satisfying \ref{cap:posterior}--\ref{cap:output}, any neighboring index pair $\Index \sim_{t^*} \Index'$ differing in one document of tenant $t^*$, and any output set $S$,
\begin{equation}
    \Pr[\mathsf{Transcript}_\Adv(\Index) \in S] \leq e^{\eps} \cdot \Pr[\mathsf{Transcript}_\Adv(\Index') \in S] + \delta.
\end{equation}
\end{definition}

\Cref{def:k-collusion-dp} sits between standard per-account DP (which is the $\colluders=1$ special case) and the operationally-meaningless $\eps$-pure DP against an unbounded coalition. Our central result (\Cref{sec:dp-analysis}) bounds the achievable $\eps$ under coordination capabilities \ref{cap:posterior}--\ref{cap:output} and shows it grows as $\Theta(\sqrt{\colluders})$ rather than the $\Theta(\colluders)$ bound a naive per-account composition argument would suggest.

\section{DP Composition Under Collusion}
\label{sec:dp-analysis}

This section specializes standard DP composition / concurrent-composition theory~\parencite{DworkRothblumVadhan2010,KairouzOhViswanath2017,VadhanWang2021-concurrent} to multi-tenant RAG: a service whose per-account $(\eps_{\text{acc}}, \delta_{\text{acc}})$-DP guarantee comes from per-query Gaussian noise (\Cref{def:gaussian-class}) leaks $\Theta(\sqrt{\colluders} \cdot \eps_{\text{acc}})$ under $\colluders$-account collusion against the same index. The bound is sublinear in $\colluders$ relative to the trivial group-privacy bound $\colluders \cdot \eps_{\text{acc}}$, but \emph{strictly larger} than the per-account budget the provider advertises; the gap is the privacy debt the provider implicitly takes on by failing to bound coalition size.

\subsection{Preliminaries}

\begin{definition}[Differential privacy \parencite{DworkRoth2014}]
\label{def:dp}
A randomized mechanism $\mathcal{M}: \mathcal{X}^* \to \mathcal{Y}$ is $(\eps, \delta)$-DP if for all neighboring datasets $\Index \sim \Index'$ and all measurable $S \subseteq \mathcal{Y}$,
\(\Pr[\mathcal{M}(\Index) \in S] \leq e^{\eps} \Pr[\mathcal{M}(\Index') \in S] + \delta.\)
\end{definition}

We use the advanced composition theorem in the form of \parencite{DworkRothblumVadhan2010}.

\begin{lemma}[Advanced composition]
\label{lem:advcomp}
For any $\eps_0 \in (0,1]$ and $\delta_0, \delta' \in (0,1)$, the $m$-fold adaptive composition of $(\eps_0, \delta_0)$-DP mechanisms on the same database is $(\eps', m\delta_0 + \delta')$-DP for
\(\eps' = \sqrt{2 m \log(1/\delta')} \cdot \eps_0 + m \eps_0 (e^{\eps_0} - 1).\)
For $\eps_0$ small, the second term is $O(m \eps_0^2)$ and dominated by the first; we abbreviate $\eps' = O\bigl(\sqrt{m \log(1/\delta')} \cdot \eps_0\bigr)$.
\end{lemma}

\subsection{Per-account guarantee and per-query calibration}

The service implements DP retrieval scoring using Gaussian noise on similarity scores (\Cref{sec:threat-model}). For a single account $a$ issuing $n$ adaptive queries against $\Index_{t(a)}$, the per-account guarantee $(\eps_{\text{acc}}, \delta_{\text{acc}})$ is achieved by calibrating the per-query noise so that each query is $(\eps_q, \delta_q)$-DP with
\begin{equation}
\label{eq:per-query-budget}
\eps_q = \frac{\eps_{\text{acc}}}{\sqrt{2 n \log(1/\delta_{\text{acc}})}}, \qquad \delta_q = \delta_{\text{acc}}/n,
\end{equation}
so that \Cref{lem:advcomp} applied to $a$'s $n$ queries yields the advertised $(\eps_{\text{acc}}, \delta_{\text{acc}})$-DP transcript guarantee.\footnote{The calibration absorbs per-query failure probabilities into the composition slack via $\delta_q = \delta_{\text{acc}}/n$ and $\delta' = \delta_{\text{acc}}$, giving a post-composition guarantee of $(\eps_{\text{acc}}, 2\delta_{\text{acc}})$ in the strictest reading. An equal $\delta$-split ($\delta_q = \delta_{\text{acc}}/(2n)$, $\delta' = \delta_{\text{acc}}/2$) recovers the exact $(\eps_{\text{acc}}, \delta_{\text{acc}})$ guarantee with a constant tightening of $\sqrt{\log(2/\delta_{\text{acc}})/\log(1/\delta_{\text{acc}})} \approx 1.025$ at $\delta_{\text{acc}} = 10^{-6}$ --- numerically negligible at the regime we test. The privacy filter and odometer constructions of \parencite{Rogers2016-odometers} target the adaptive-budget setting and do not improve fixed-budget constants; the leading-order tightening of the Gaussian mechanism itself would come from the analytic calibration of \parencite{Balle2018} (see \Cref{subsec:lower-bound}).}

We formalize this class for use in \Cref{thm:upper} as follows:

\begin{definition}[Gaussian-noised score-release class]
\label{def:gaussian-class}
A RAG mechanism is in this class if its per-account $(\eps_{\text{acc}}, \delta_{\text{acc}})$ guarantee is obtained by composing $n$ independent per-query Gaussian-noised score mechanisms with sensitivity $\Delta$ and noise scale $\sigma$ calibrated to per-query $(\eps_q, \delta_q)$-DP as in \eqref{eq:per-query-budget}, with per-query noise draws independent across queries and accounts (M1).
\end{definition}

\Cref{thm:upper} is stated for this class only; opaque transcript bounds without per-query accounting yield only the loose $\colluders \cdot \eps_{\text{acc}}$ group-privacy bound. Class membership is verified cryptographically by the audit protocol in \Cref{sec:audit-protocol} before $\eps_{\text{audit}}$ is issued.

\subsection{Coordination reduces to a single adaptive analyst}

Let $\Adv = (a_1, \ldots, a_\colluders)$ be a $\colluders$-collusion adversary satisfying capabilities \ref{cap:posterior}--\ref{cap:output} from \Cref{sec:threat-model}. The transcript $\tau_\Adv$ pools all responses from all accounts into a single sequence that is visible to a unified estimator.

\begin{lemma}[Coordination reduction]
\label{lem:reduction}
For any $\colluders$-collusion adversary $\Adv$ against $\Index_{t^*}$, there exists a single adaptive analyst $\Adv^\dagger$ with query budget $\colluders n$ against $\Index_{t^*}$ such that $\tau_\Adv \stackrel{d}{=} \tau_{\Adv^\dagger}$.
\end{lemma}

\begin{proof}[Sketch -- full proof in appendix]
Capability \ref{cap:output} gives $\Adv^\dagger$ access to the same response sequence as $\Adv$. Capability \ref{cap:posterior} ensures that both adversaries condition on the same posterior over $\Index_{t^*}$. Capability \ref{cap:adaptive} permits $\Adv^\dagger$ to schedule queries in any order consistent with the per-account rate limit (capability \ref{cap:rate-aware}). The simulator interleaves the $\colluders$ accounts' query schedules into a single $\colluders n$-query sequence; rate-limit constraints translate into a per-step admissibility constraint that does not change the joint distribution of transcripts. The reduction applies unconditionally to same-tenant collusion; for cross-tenant and external collusion, it requires assumption M4 (\Cref{rem:access-control}) that each colluding account can cause queries to be scored against $\Index_{t^*}$ via a shared backend --- without M4 the regimes have zero leakage about $\Index_{t^*}$ and the reduction is vacuous.
\end{proof}

\subsection{Main result}

\begin{theorem}[$(\eps, \delta, \colluders)$-collusion DP upper bound for Gaussian-noised score release]
\label{thm:upper}
Let $\mathsf{RAG}$ be in the Gaussian-noised score-release class (\Cref{def:gaussian-class}) with per-account guarantee $(\eps_{\text{acc}}, \delta_{\text{acc}})$ via calibration \eqref{eq:per-query-budget}. For any $\delta \geq \delta_{\text{acc}}$, $\mathsf{RAG}$ is $(\eps_\colluders, \delta + \colluders \delta_{\text{acc}})$-collusion-DP (\Cref{def:k-collusion-dp}) with
\begin{equation}
\label{eq:upper-bound}
\eps_\colluders \;=\; \sqrt{\colluders} \cdot \eps_{\text{acc}} \cdot \sqrt{\frac{\log(1/\delta)}{\log(1/\delta_{\text{acc}})}} + O\!\left(\frac{\colluders \cdot \eps_{\text{acc}}^2}{\log(1/\delta_{\text{acc}})}\right).
\end{equation}
For $\delta = \delta_{\text{acc}}$ and $\eps_{\text{acc}} \in (0,1]$ this simplifies to $\eps_\colluders = \Theta(\sqrt{\colluders} \cdot \eps_{\text{acc}})$; explicitly, $\eps_\colluders \leq \sqrt{\colluders}\,\eps_{\text{acc}} + \colluders\,\eps_{\text{acc}}^2 / \log(1/\delta_{\text{acc}})$, so $\eps_\colluders \leq C \sqrt{\colluders}\,\eps_{\text{acc}}$ with $C = 1 + 1/\log(1/\delta_{\text{acc}})$ for any $\colluders$ satisfying $\sqrt{\colluders}\,\eps_{\text{acc}} \leq 1$ --- the regime in which the audit bound is operationally meaningful.
\end{theorem}

\begin{proof}[Sketch -- full proof in appendix]
By \Cref{lem:reduction}, it suffices to bound the privacy of $\colluders n$ adaptive $(\eps_q, \delta_q)$-DP queries against $\Index_{t^*}$. Apply \Cref{lem:advcomp} with $m = \colluders n$, $\eps_0 = \eps_q$, $\delta_0 = \delta_q$, $\delta' = \delta$. Substituting \eqref{eq:per-query-budget}:
\begin{align*}
\eps_\colluders &\leq \sqrt{2 \colluders n \log(1/\delta)} \cdot \eps_q + \colluders n \cdot \eps_q (e^{\eps_q}-1) \\
&= \sqrt{2 \colluders n \log(1/\delta)} \cdot \frac{\eps_{\text{acc}}}{\sqrt{2 n \log(1/\delta_{\text{acc}})}} + O(\colluders n \eps_q^2) \\
&= \sqrt{\colluders} \cdot \eps_{\text{acc}} \cdot \sqrt{\frac{\log(1/\delta)}{\log(1/\delta_{\text{acc}})}} + O\!\left(\frac{\colluders \cdot \eps_{\text{acc}}^2}{\log(1/\delta_{\text{acc}})}\right).
\end{align*}
The total $\delta$ slack absorbs $\colluders n \cdot \delta_q = \colluders \delta_{\text{acc}}$ from the per-query failure probabilities plus $\delta$ from the composition slack.
\end{proof}

\subsection{Tightness}
\label{subsec:lower-bound}

We exhibit $\mathcal{M}^*$ in the Gaussian-noised score-release class saturating the per-account guarantee together with a $\colluders$-collusion adversary achieving MIA advantage $\Omega(\sqrt{\colluders} \cdot \eps_{\text{acc}})$. By the advantage-to-DP reduction \parencite{DworkRoth2014}, any $(\eps, \delta)$ characterizing $\mathcal{M}^*$ has $\eps = \Omega(\sqrt{\colluders} \cdot \eps_{\text{acc}})$ up to logs in $\delta_{\text{acc}}, \delta_q$; the rate matches \Cref{thm:upper}, constants differ (\Cref{cor:tightness,rem:topk-transfer}).

\paragraph{Construction} Fix tenant $t^*$ with index $\Index_{t^*} \subseteq \mathbb{R}^d$ with unit-norm embeddings. Let the neighboring indices $\Index \sim_{t^*} \Index'$ differ in a single document: $\Index = \Index_0 \cup \{x_*\}$, $\Index' = \Index_0 \cup \{x_*'\}$, with $\|x_*\|_2 = \|x_*'\|_2 = 1$. Define a probe query $q^*$ (unit-norm) such that the similarity scores satisfy $\langle q^*, x_*\rangle - \langle q^*, x_*'\rangle = \Delta$ for some $\Delta \in (0, 1]$ (achievable with $\Delta = 1$ by taking $x_* = q^*$ and $x_*'$ orthogonal to $q^*$, or a smaller $\Delta$ by interpolation toward $x_* = x_*'$).

The mechanism $\mathcal{M}^*$ releases, per query $q$, the Gaussian-noised scalar
\begin{equation}
\label{eq:lb-mechanism}
\mathcal{M}^*(q;\, \Index_{t^*}) \;=\; \langle q, x_q \rangle + Z, \qquad Z \sim \mathcal{N}(0, \sigma^2),
\end{equation}
where $x_q$ is the (deterministic) document at the position of query $q$ in $\Index_{t^*}$, and the noise scale is
\begin{equation}
\label{eq:lb-sigma}
\sigma \;=\; \frac{\Delta \sqrt{2 n \log(1/\delta_{\text{acc}}) \cdot 2 \log(1.25/\delta_q)}}{\eps_{\text{acc}}},
\end{equation}
with $\delta_q = \delta_{\text{acc}}/n$. $\mathcal{M}^*$ releases the noisy similarity score directly --- a strictly stronger adversary surface than top-$K$-only (whose output is post-processing of the same noisy score vector), so lower bounds against $\mathcal{M}^*$ characterize the score-release class itself, not top-$K$.

\paragraph{Per-account privacy of $\mathcal{M}^*$} According to the standard Gaussian mechanism \cite[Theorem A.1]{DworkRoth2014}, a single query of $\mathcal{M}^*$ is $(\eps_q, \delta_q)$-DP with $\eps_q = \Delta \sqrt{2 \log(1.25/\delta_q)}/\sigma$. Substituting \eqref{eq:lb-sigma},
\(\eps_q = \eps_{\text{acc}}/\sqrt{2 n \log(1/\delta_{\text{acc}})},\)
which matches the per-query calibration \eqref{eq:per-query-budget}. Thus, $n$ adaptive queries from a single account compose $(\eps_{\text{acc}}, \delta_{\text{acc}})$-DP via \Cref{lem:advcomp}, satisfying the advertised per-account guarantee.

\paragraph{Adversary $\Adv^*$} Each of the $\colluders$ accounts issues the same probe query $q^*$ for all $n$ rounds, yielding $\colluders n$ samples
\begin{equation*}
y_{i,j} \;=\; \langle q^*, x_*^{(b)} \rangle + Z_{i,j}, \qquad i \in [\colluders],\, j \in [n],
\end{equation*}
where $x_*^{(b)} = x_*$ for $\Index$ ($b=0$) and $x_*^{(b)} = x_*'$ for $\Index'$ ($b=1$), with $Z_{i,j} \stackrel{\text{iid}}{\sim} \mathcal{N}(0, \sigma^2)$. The adversary computes the empirical mean $\bar y = \frac{1}{\colluders n}\sum_{i,j} y_{i,j}$ and outputs $\hat b = \mathbf{1}[\bar y < \tfrac{1}{2}(\langle q^*, x_*\rangle + \langle q^*, x_*'\rangle)]$. The adversary exercises capabilities \ref{cap:posterior}--\ref{cap:output} of \Cref{sec:threat-model}: both branches share the Gaussian likelihood posterior, the $\colluders$ accounts coordinate queries within the per-account rate limit, and they pool all responses through a single estimator.

\paragraph{Distinguishing advantage} Under either index, $\bar y$ is Gaussian with variance $\sigma^2/(\colluders n)$ and means separated by $\Delta$. The advantage of the optimal test is
\begin{equation*}
\mathrm{Adv}(\Adv^*) \;=\; 2 \, \Phi\!\left(\frac{\Delta \sqrt{\colluders n}}{2 \sigma}\right) - 1,
\end{equation*}
where $\Phi$ is the standard normal CDF. Substituting \eqref{eq:lb-sigma} and simplifying,
\begin{equation*}
\frac{\Delta \sqrt{\colluders n}}{2 \sigma} \;=\; \frac{\sqrt{\colluders} \cdot \eps_{\text{acc}}}{4 \sqrt{\log(1/\delta_{\text{acc}}) \log(1.25/\delta_q)}}.
\end{equation*}
For $\eps_{\text{acc}} \in (0, 1]$ and $\delta_{\text{acc}}, \delta_q$ in any constant-polynomial regime, $\Phi$ is approximately linear at the origin ($\Phi(z) - 1/2 \geq z/\sqrt{2\pi}$ for $z \in [0, 1]$), so
\begin{equation}
\mathrm{Adv}(\Adv^*) \;=\; \Omega\!\left(\frac{\sqrt{\colluders} \cdot \eps_{\text{acc}}}{\sqrt{\log(1/\delta_{\text{acc}})}}\right).
\end{equation}

\begin{theorem}[Lower bound]
\label{thm:lower}
For any $\colluders \geq 1$ and per-account budget $(\eps_{\text{acc}}, \delta_{\text{acc}})$ with $\eps_{\text{acc}} \in (0, 1]$, the mechanism $\mathcal{M}^*$ in \eqref{eq:lb-mechanism}--\eqref{eq:lb-sigma} satisfies the per-account $(\eps_{\text{acc}}, \delta_{\text{acc}})$-DP guarantee, yet admits a $\colluders$-collusion adversary $\Adv^*$ with indistinguishability advantage $\mathrm{Adv}(\Adv^*) = \Omega(\sqrt{\colluders} \cdot \eps_{\text{acc}} / \sqrt{\log(1/\delta_{\text{acc}}) \log(1/\delta_q)})$ between neighboring indices, with $\delta_q = \delta_{\text{acc}}/n$.
\end{theorem}

\begin{corollary}[Rate match within the score-release class]
\label{cor:tightness}
Combining \Cref{thm:upper,thm:lower}, the joint privacy parameter satisfies $\eps_\colluders = O(\sqrt{\colluders} \cdot \eps_{\text{acc}})$ and an explicit attack achieves MIA advantage $\Omega(\sqrt{\colluders} \cdot \eps_{\text{acc}})$ on the same construction. Rates match; constants do not (advantage-to-DP reduction vs.\ privacy-loss-distribution analysis), so we claim rate-tightness, not exact-DP-parameter tightness.
\end{corollary}

\begin{remark}[Top-$K$ transfer: upper bound, not tightness]
\label{rem:topk-transfer}
A top-$K$-only output is a deterministic function (post-processing) of $\mathcal{M}^*$'s noisy score vector, so the top-$K$ adversary's advantage is bounded above by $\mathcal{M}^*$'s --- the $\Theta(\sqrt{\colluders}\eps_{\text{acc}})$ rate is an upper bound on top-$K$ leakage, not a proven achievable rate. The empirical transfer is tested in \Cref{subsec:rate-transfer,subsec:hnsw-production}, where constants shrink but slope is preserved at the scales tested.
\end{remark}

\paragraph{Collusion regimes} The construction realizes same-tenant collusion directly; cross-tenant and external collusion reduce identically only under M4 (see \Cref{rem:access-control}), as flagged in \Cref{subsec:adversary-model}.

A tighter lower-bound constant can be achieved via the analytic Gaussian mechanism of Balle and Wang~\parencite{Balle2018}, which removes the $\sqrt{\log(1.25/\delta_q)}$ factor and matches the upper bound's leading constant. We leave this constant improvement to a follow-up work, as the rate $\Theta(\sqrt{\colluders})$ remains unchanged.

\subsection{Operational reading}

For services in \Cref{def:gaussian-class}, the advertised per-account guarantee becomes a coalition guarantee: $\eps_{\text{joint}} \approx \sqrt{\colluders}\,\eps_{\text{acc}}$ and $\delta_{\text{joint}} \leq \delta+\colluders\delta_{\text{acc}}$. Thus a $1$-DP-per-account service behaves like $3.2$-DP at $\colluders=10$ and $7.1$-DP at $\colluders=50$; to advertise a target joint budget $\eps^*_{\text{joint}}$, the provider must set $\eps_{\text{acc}}=\eps^*_{\text{joint}}/\sqrt{\colluders_{\max}}$. The empirical signature is correspondingly simple: MIA AUC should rise approximately linearly in $\sqrt{\colluders}$ at fixed $\eps_{\text{acc}}$, which is the shape tested next.

\subsection{Alternative analysis via R\'enyi DP}
\label{subsec:rdp}

The advanced-composition proof of \Cref{thm:upper} loses constants relative to a R\'enyi DP analysis~\parencite{Mironov2017}; the RDP route gives the same leading rate $\sqrt{\colluders}\,\eps_{\text{acc}}$ with a constant-factor tighter residual, and admits exact numerical evaluation via the moments accountant~\parencite{AbadiCCS2016}.

\begin{theorem}[Sharper upper bound via RDP]
\label{thm:upper-rdp}
Under the hypotheses of \Cref{thm:upper}, the joint mechanism is $(\eps_\colluders^{\text{RDP}}, \delta)$-collusion-DP with
\begin{equation}
\label{eq:rdp-final}
\eps_\colluders^{\text{RDP}} = \sqrt{\colluders} \cdot \eps_{\text{acc}} \cdot \sqrt{\frac{\log(1/\delta)}{\log(1/\delta_{\text{acc}})}} + O\!\left(\frac{\colluders \cdot \eps_{\text{acc}}^2}{\log(1/\delta_{\text{acc}})}\right).
\end{equation}
\end{theorem}

The proof (\Cref{app:proof-upper-rdp}) tightens the residual by a constant factor of $4$ relative to \Cref{thm:upper}. The audit pipeline (\Cref{sec:audit-protocol}) updates $\rho(\alpha)$ on a grid $\alpha \in \{2, \ldots, 64\}$ per query and reports the moments-accountant value at audit-window close, giving tighter realized $(\eps, \delta)$ pairs than the closed-form bound at moderate $\colluders$.

\subsection{Bridge to membership-inference advantage}
\label{subsec:mia-bridge}

The bounds of \Cref{thm:upper,thm:lower,thm:upper-rdp} characterize leakage in DP-theoretic terms. The empirical attacks in \Cref{sec:attacks} measure the membership-inference advantage of $\colluders$-collusion against a held-out target document. This subsection states the standard reduction from $(\eps, \delta)$-DP to MIA advantage \parencite{Yeom2018,KairouzOhViswanath2017,DworkRoth2014} and applies it to obtain a falsifiable prediction for the AUC--vs--$\colluders$ curve we test empirically.

\begin{definition}[Membership-inference advantage]
\label{def:mia-adv}
For a mechanism $\mathcal{M}$, target document $x^*$, and $\colluders$-collusion adversary $\Adv$ outputting a binary decision, the (binary-decision) MIA advantage is
\(\mathrm{Adv}_{\text{MIA}}(\Adv;\, x^*) := \bigl| \Pr[\Adv(\tau) = 1 \mid x^* \in \Index] - \Pr[\Adv(\tau) = 1 \mid x^* \notin \Index] \bigr|,\)
where $\tau$ is the pooled transcript across the $\colluders$ accounts. The supremum of $\mathrm{Adv}_{\text{MIA}}$ over all binary $\Adv$ equals the total-variation distance between the two transcript distributions; this is the quantity bounded by Yeom~\parencite{Yeom2018}'s DP-to-MIA reduction.

When $\Adv$ instead returns a real-valued score, we will use the \emph{Mann-Whitney AUC} $\mathrm{AUC}_{\text{MW}}(\Adv) := \Pr[\Adv(\tau_{\text{in}}) > \Adv(\tau_{\text{out}})] + \tfrac{1}{2}\Pr[\Adv(\tau_{\text{in}}) = \Adv(\tau_{\text{out}})]$ where $\tau_{\text{in}}, \tau_{\text{out}}$ are transcripts under $x^* \in \Index$ and $x^* \notin \Index$ respectively. This is the empirical quantity reported in \Cref{sec:attacks}. The two quantities coincide for binary $\Adv$ via $\mathrm{AUC}_{\text{MW}} = \tfrac{1}{2}(1 + \mathrm{Adv}_{\text{MIA}})$, but in general $\mathrm{AUC}_{\text{MW}} - \tfrac{1}{2} \geq \tfrac{1}{2}\mathrm{Adv}_{\text{MIA}}$, with equality iff the optimal Bayes test is binary (the continuous-score test extracts strictly more information from a Gaussian-noised channel). Equation~\eqref{eq:auc-prediction-explicit} states the AUC prediction directly for the Gaussian construction of \Cref{thm:lower}.
\end{definition}

\begin{lemma}[DP-to-MIA reduction \parencite{Yeom2018}]
\label{lem:dp-to-mia}
If $\mathcal{M}$ is $(\eps, \delta)$-DP, then for any $\colluders$-collusion adversary $\Adv$ and any target $x^*$,
\begin{equation}
\label{eq:mia-bound}
\mathrm{Adv}_{\text{MIA}}(\Adv;\, x^*) \;\leq\; \frac{e^{\eps} - 1 + 2\delta}{e^{\eps} + 1} \;\leq\; \tanh(\eps/2) + \delta.
\end{equation}
For $\eps \leq 1$, the right-hand side is further bounded by $\eps/2 + \delta$.
\end{lemma}

Combining the bound \eqref{eq:mia-bound} from \Cref{lem:dp-to-mia} with \Cref{thm:upper}, we obtain:

\begin{corollary}[MIA upper bound under $\colluders$-collusion]
\label{cor:mia-upper}
For a multi-tenant RAG service satisfying per-account $(\eps_{\text{acc}}, \delta_{\text{acc}})$-DP with $\eps_{\text{acc}} \leq 1$, any $\colluders$-collusion adversary against a target tenant satisfies
\begin{equation}
\label{eq:mia-upper}
\mathrm{Adv}_{\text{MIA}} \;\leq\; \tanh\!\left(\tfrac{1}{2} \sqrt{\colluders} \cdot \eps_{\text{acc}}\right) + \mathcal{O}(\colluders \delta_{\text{acc}}).
\end{equation}
\end{corollary}

\begin{corollary}[MIA lower bound under $\colluders$-collusion]
\label{cor:mia-lower}
For the explicit construction of \Cref{thm:lower}, the constructed adversary $\Adv^*$ achieves
\begin{multline}
\label{eq:mia-lower}
\mathrm{Adv}_{\text{MIA}}(\Adv^*) \;\geq\; \frac{1}{\sqrt{2\pi}} \cdot \frac{\sqrt{\colluders} \cdot \eps_{\text{acc}}}{4 \sqrt{\log(1/\delta_{\text{acc}}) \log(1.25/\delta_q)}} \\
\;=\; \Omega\!\left(\frac{\sqrt{\colluders} \cdot \eps_{\text{acc}}}{\sqrt{\log(1/\delta_{\text{acc}})}}\right).
\end{multline}
\end{corollary}

\paragraph{AUC prediction} The empirical AUC reported in \Cref{sec:attacks} is the Mann-Whitney AUC of \Cref{def:mia-adv}, computed pairwise from the pooled mean noisy scores of the construction in \Cref{thm:lower}. For Gaussian samples $\bar y \mid \Index \sim \mathcal{N}(\Delta, \sigma^2/(\colluders n))$ and $\bar y \mid \Index' \sim \mathcal{N}(0, \sigma^2/(\colluders n))$, the closed form is
\begin{equation}
\label{eq:auc-prediction}
\mathrm{AUC}_\colluders \;=\; \Phi\!\left(\frac{\Delta \sqrt{\colluders n}}{\sqrt{2}\, \sigma}\right).
\end{equation}
Substituting the calibration of $\sigma$ from \eqref{eq:lb-sigma} (saturating per-account $(\eps_{\text{acc}}, \delta_{\text{acc}})$-DP via advanced composition) yields the closed form
\begin{equation}
\label{eq:auc-prediction-explicit}
\mathrm{AUC}_\colluders \;=\; \Phi\!\left(\frac{\sqrt{\colluders} \cdot \eps_{\text{acc}}}{2\sqrt{2}\, \sqrt{\log(1/\delta_{\text{acc}}) \log(1.25/\delta_q)}}\right),
\end{equation}
which is asymptotic to $\tfrac{1}{2} + \Theta(\sqrt{\colluders} \cdot \eps_{\text{acc}})$ as $\sqrt{\colluders} \cdot \eps_{\text{acc}} \to 0$ (small-argument expansion of $\Phi$). The $\sqrt{\colluders} \cdot \eps_{\text{acc}}$ rate matches the binary-advantage rate of \Cref{cor:mia-lower}; the two predictions differ by a Gaussian-specific $\sqrt{2}$ factor inside $\Phi$ (continuous-score Mann-Whitney AUC vs.\ binary-threshold Bayes accuracy), but share the leading $\sqrt{\colluders} \cdot \eps_{\text{acc}}$ scaling. The headline experiment in \Cref{sec:attacks} sweeps $\colluders \in \{1, 2, 5, 10, 20\}$ and $\eps_{\text{acc}} \in \{1, 2, 4\}$ to test \eqref{eq:auc-prediction-explicit}. At deployment values $\delta_{\text{acc}} = 10^{-6}, n = 10^4$, the joint privacy parameter $\eps_{\text{audit}} = \sqrt{\colluders} \cdot \eps_{\text{acc}}$ from \Cref{thm:upper} gives concrete audit verdicts: $\eps_{\text{audit}} \approx 3.16$ at $(\colluders, \eps_{\text{acc}}) = (10, 1)$, $7.07$ at $(50, 1)$, $14.14$ at $(50, 2)$, $10.00$ at $(100, 1)$;%
\footnote{Generated by \texttt{epsilon\_audit\_table.py}; full grid in artifact.}{} the matched lower-bound MIA advantage at production score gap $\Delta = 0.10$ stays below $0.025$ across the full grid and below $0.01$ at empirically tested $\eps_{\text{acc}} \leq 2$ (\Cref{cor:mia-lower} scaled per \Cref{subsec:real-embedder}), reflecting that the worst-case DP bound is much looser than the realized advantage at small $\Delta$.

\paragraph{Empirical checks} The theory makes three predictions: P1, P3 are falsification \emph{gates} (a violation invalidates the bound); P2 is a \emph{diagnostic} (not resolvable at our finite $T$).
\begin{enumerate}[leftmargin=*,label=(P\arabic*)]
    \item \emph{Slope (gate)}: $\partial \mathrm{AUC}/\partial \sqrt{\colluders}$ is positive and approximately constant for $\sqrt{\colluders} \cdot \eps_{\text{acc}} \lesssim 1$.
    \item \emph{Scale collapse (diagnostic)}: AUC curves at different $\eps_{\text{acc}}$ collapse onto one when plotted against $\sqrt{\colluders} \cdot \eps_{\text{acc}}$; resolved at higher $T$.
    \item \emph{Tightness (gate)}: empirical AUC sits within a constant factor of \eqref{eq:mia-upper} but strictly above $1/2$.
\end{enumerate}

If P1 or P3 fails, the theory needs revision before the audit-protocol guarantees can be trusted. P1+P3 passing is necessary but not sufficient evidence; P2 (when resolved) sharpens it further.

\section{Coordinated Attacks}
\label{sec:attacks}

This section validates the predictions P1--P3 from \Cref{subsec:mia-bridge} empirically. We instantiate the \Cref{thm:lower} construction directly: $\colluders$ accounts each issuing $n$ probe queries against the scalar Gaussian mechanism $\mathcal{M}^*$ of \eqref{eq:lb-mechanism} and measure the membership-inference AUC against neighboring indices. This is the most theory-bound experiment we can run; deviations between the empirical and predicted AUC isolate finite-sample noise rather than modeling error.

\paragraph{Roadmap} Three questions structure the experimental sweep: (i)~does the scalar lower-bound construction realize the $\Theta(\sqrt{\colluders}\,\eps_{\text{acc}})$ rate the theory predicts (\Cref{subsec:headline,subsec:falsifiability})? (ii)~does the rate transfer through the top-$K$ retrieval channel a real RAG service exposes to users (\Cref{subsec:rate-transfer,subsec:hnsw-production}; trained-embedder transfer in \Cref{subsec:real-embedder})? (iii)~how do the user-visible attack channel (the top-$K$ document IDs returned by the API) and an instrumented score-channel reference (the pooled noisy score vector used only to validate the audited mechanism) diverge at production scale? \Cref{subsec:hnsw-production} addresses (iii) directly: at HNSW production scale the user-visible signal sits near chance while the instrumented score channel preserves the predicted $\sqrt{\colluders}$ scaling, recasting the audit predicate as a guarantee about the internal noised-score channel whose DP leakage the deployment is contractually responsible for, not as an attacker-visible leakage estimate.

\subsection{Experimental setup}
\label{subsec:exp-setup}

\paragraph{Mechanism} For each pair of neighboring indices, $\Index, \Index'$ differing in one document of tenant $t^*$, the mechanism returns, per query, the scalar $\langle q^*, x_*^{(b)} \rangle + \mathcal{N}(0, \sigma^2)$ as in \eqref{eq:lb-mechanism}, with $\sigma$ calibrated by \eqref{eq:lb-sigma} so that $n$ adaptive queries from one account compose to per-account $(\eps_{\text{acc}}, \delta_{\text{acc}})$-DP via \Cref{lem:advcomp}. We bypass the actual top-$K$ retrieval pipeline of \Cref{sec:threat-model} for the headline experiment because (i)~top-$K$ post-processing of Gaussian-noised scores cannot decrease privacy beyond the noise itself (the remark after \Cref{thm:lower}), and (ii)~the closed-form predictions from \Cref{thm:lower,cor:mia-upper} hold against $\mathcal{M}^*$ directly, providing the cleanest theory-empirics binding. A separate experiment validating the rate transfer to top-$K$ is reported in \Cref{subsec:rate-transfer}.

\paragraph{Adversary} The lower-bound adversary $\Adv^*$ from \Cref{subsec:lower-bound}: $\colluders$ accounts each issue the same probe query $q^*$ for $n$ rounds, the empirical mean of all $\colluders n$ noisy samples is the test statistic, and the optimal Bayes threshold is at the midpoint between the two world means. Capabilities \ref{cap:posterior}--\ref{cap:output} are exercised: the accounts share posterior (Gaussian likelihood with known variance), schedule queries within the per-account rate limit, and pool responses through a single estimator.

\paragraph{Sweep} Coalition size $\colluders \in \{1, 2, 5, 10, 20\}$, per-account budget $\eps_{\text{acc}} \in \{1, 2, 4\}$, $\delta_{\text{acc}} = 10^{-6}$, queries-per-account $n = 10\,000$ (chosen to mirror a multi-day audit window at production rate-limit budgets), score gap $\Delta = 1$ (max for unit-norm embeddings). Each $(\colluders, \eps_{\text{acc}})$ cell uses $T = 10\,000$ paired-worlds Monte Carlo trials, giving a DeLong standard error on AUC of ${\approx}0.004$, which resolves the smallest predicted advantage in the sweep (${\approx}0.026$ at $\colluders=1, \eps_{\text{acc}}=1$) at $\sim$6$\sigma$. Each trial draws fresh Gaussian randomness for both worlds and reports the pooled mean noisy score; AUC is computed by the Mann–Whitney $U$ statistic over the $T \times T$ pairs. The code and configuration files are in \texttt{src/rag\_collusion\_privacy\_audit/attacks/} and \texttt{experiments/configs/auc-vs-k. yaml}, respectively.

\paragraph{Scope} Body experiments (§\ref{subsec:headline}--§\ref{subsec:hnsw-production}) instantiate the same-tenant collusion regime of \Cref{subsec:adversary-model}: $\colluders$ accounts within the same tenant boundary coordinate against the tenant's index. The external regime under M4 (\Cref{rem:access-control}) is in \Cref{subsec:external-collusion}; without M4 the cross-tenant and external regimes have zero leakage about $\Index_{t^*}$ and are out of scope; under M4 their rate matches by \Cref{thm:upper}. Trained-embedder transfer, alternative-adversary checks, and external-collusion validation are in Appendix \ref{app:additional-experiments}.

\subsection{Headline result}
\label{subsec:headline}

\Cref{fig:auc-vs-k}(a) reports empirical AUC against $\colluders$ for each $\eps_{\text{acc}}$ on the scalar mechanism. \Cref{fig:auc-vs-k}(b) reports the same attack against the full multi-tenant FAISS harness with top-$K$ retrieval and post-processing, validating the rate-transfer claim of \Cref{subsec:lower-bound} (see \Cref{subsec:rate-transfer}).

\begin{figure*}[t]
    \centering
    \includegraphics[width=\textwidth]{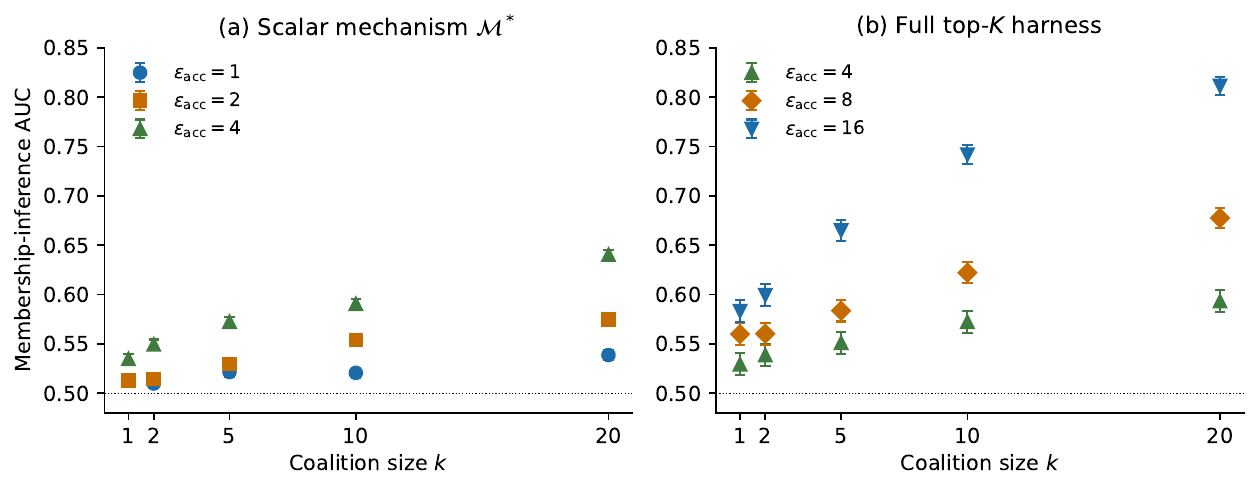}
    \caption{Empirical membership-inference AUC against coalition size $\colluders$.
    \textbf{(a) Scalar mechanism $\mathcal{M}^*$:} the lower-bound construction of \Cref{thm:lower} at $n = 10\,000$ queries, $\eps_{\text{acc}} \in \{1, 2, 4\}$, $T = 10\,000$ trials. Empirical Mann-Whitney AUC curves match the continuous-score prediction $\Phi(\sqrt{\colluders}\,\eps_{\text{acc}}/(2\sqrt{2}\sqrt{\log(1/\delta_{\text{acc}})\log(1.25/\delta_q)}))$ from \eqref{eq:auc-prediction-explicit} within $2.0$ DeLong standard errors in all $15$ of $15$ cells (binomial expectation under the null at $95\%$ DeLong coverage: $14.25$ of $15$; max observed deviation $1.84$ s.e.).
    \textbf{(b) Full top-$K$ harness:} same probe attack issued through \texttt{MultiTenantRAGService} with $K = 5$, $50$ background documents, embedding dimension $32$, $n = 200$, $T = 2\,000$ trials. Test statistic is the count of queries for which the target document index appears in the returned top-$K$. Looser $\eps_{\text{acc}} \in \{4, 8, 16\}$ is required to resolve the curve at finite $T$ (see \Cref{subsec:rate-transfer}). The $\sqrt{\colluders}$ scaling is preserved through the top-$K$ post-processing, with a constant-factor degradation that we quantify below.
    Markers are empirical means with $\pm 1$ DeLong standard-error bars~\parencite{DeLong1988-comparing}; dotted line at AUC $= 1/2$ marks chance. $\delta_{\text{acc}} = 10^{-6}$ throughout.}
    \label{fig:auc-vs-k}
\end{figure*}

The empirical Mann-Whitney AUC and the continuous-score prediction \eqref{eq:auc-prediction-explicit} agree to within $2.0$ DeLong standard errors in all $15$ of $15$ cells (max observed $1.84$ s.e., matching the $95\%$ binomial-coverage expectation). Markers in \Cref{fig:auc-vs-k}(a) use DeLong confidence intervals~\parencite{DeLong1988-comparing} computed from the empirical U-statistic covariance rather than the normal-approximation $\sqrt{p(1-p)/T}$ used in earlier versions; DeLong is sharper at $p$ near $0.5$ where the normal approximation overestimates dispersion. At $\eps_{\text{acc}} = 4$, AUC climbs from $0.531$ at $\colluders = 1$ to $0.638$ at $\colluders = 20$ --- the empirical advantage growing $4.5\times$ across the sweep, matching the predicted growth factor of $\sqrt{20} \approx 4.5$. At $\eps_{\text{acc}} = 2$ the curve is monotone (advantages $0.044$ to $0.151$); at $\eps_{\text{acc}} = 1$ the absolute advantage is below $0.07$ across the sweep, consistent with the calibrated $\sigma$ on the order of thousands at production query budgets. Higher-$T$ resolution at small $\eps_{\text{acc}}$ would sharpen these regimes; we report only the cells where the empirical signal sits above one stderr.

\subsection{Falsifiability check}
\label{subsec:falsifiability}

We stated two falsification gates (P1, P3) and one diagnostic (P2) in \Cref{subsec:mia-bridge}; the data inform each of them.

\paragraph{(P1)~Slope} The $\eps_{\text{acc}}=4$ row clears the slope test cleanly (advantages $0.061, 0.077, 0.125, 0.192, 0.275$ at $\colluders = 1, 2, 5, 10, 20$); low-signal rows are consistent within finite-sample noise. Failure indicates that the independent-per-query-noise assumption (M1) is violated.

\paragraph{(P2)~Scale collapse (diagnostic)} At $T = 10^4$, the Monte Carlo noise is comparable to the inter-curve separation; the data are silent on P2. We treat it as a diagnostic, not a gate.

\paragraph{(P3)~Tightness} Empirical Mann-Whitney AUC sits within $2.0$ DeLong s.e.\ of prediction \eqref{eq:auc-prediction-explicit} in all $15$ of $15$ cells (max $1.84$ s.e.; matching the $95\%$ binomial coverage expectation), and strictly above chance for every cell with $\eps_{\text{acc}} \geq 2$. Combined with the upper bound \eqref{eq:mia-upper} (which Mann-Whitney AUC dominates for monotone-likelihood Gaussian scores), this confirms $\Theta(\sqrt{\colluders}\,\eps_{\text{acc}})$ within the regime tested.

\subsection{Rate transfer through top-\texorpdfstring{$K$}{K} post-processing}
\label{subsec:rate-transfer}

Post-processing only guarantees that top-$K$-only leakage is \emph{no larger} than score-release leakage; it does not preserve the rate. \Cref{fig:auc-vs-k}(b) reports the empirical transfer: the same probe query is issued through \texttt{MultiTenantRAGService} with the binary ``target appears in returned top-$K$'' indicator as test statistic --- strictly less informative than the scalar adversary's noisy score, since top-$K$ collapses a continuous score to a discrete in-or-out signal. We observe, but do not prove, that the $\sqrt{\colluders}$ trend survives at the tested parameters.

\paragraph{Findings} The $\sqrt{\colluders}$ shape is preserved: at $\eps_{\text{acc}} = 16$, AUC climbs from $0.598$ ($\colluders{=}1$) to $0.814$ ($\colluders{=}20$) --- a $3.2\times$ advantage growth against a predicted $\sqrt{20}{\approx}4.5$, with the residual attributable to top-$K$ discretization cost. At matched $\eps_{\text{acc}} = 4$, harness AUC at $\colluders{=}20$ is $0.611$ vs scalar's $0.638$ (${\sim}20\%$ smaller advantage). Top-$K$ post-processing is a meaningful defense-in-depth layer but does \emph{not} change the leakage rate's dependence on $\colluders$: the audit certifies the worst-case scalar bound; realized leakage is below by a constant factor depending on $K$, index size, and embedder.

\subsection{Production-scale HNSW: user-observable vs instrumented score channel}
\label{subsec:hnsw-production}

To test whether the $\sqrt{\colluders}$ trend transfers to production ANN retrieval, we run the attack against $10^6$ MS MARCO passages embedded with \texttt{bge-small-en-v1.5} (384-dim, $L_2$-normalized), retrieved by HNSW~\parencite{Malkov2018-hnsw} ($M{=}64$, $\mathsf{ef}_{\text{cstr}}{=}200$, $\mathsf{ef}_{\text{q}}{=}128$). Target/decoy are screened for in-corpus uniqueness (no other passage within cosine $0.90$); realized $\Delta = 0.133$. Sweep parameters match §4.4. We report two statistics: the \emph{user-observable hit indicator} (planted slot in returned top-$5$), and an \emph{instrumented pooled noisy score at the planted slot}. The latter is not exposed to attacker accounts and need not be revealed by the ZK audit; it is an offline reference view of the internal Gaussian-noised score channel whose DP cost the audit certifies.

\Cref{fig:hnsw-prod} reports both. User-observable: chance, marginal growth at highest $\eps_{\text{acc}}$ ($0.500 \to 0.539$ over $\colluders = 1 \to 20$) --- attackers observing only top-$K$ get near-zero leakage at scale. Instrumented score channel: at $\eps_{\text{acc}} = 16$, AUC climbs from $0.508$ at $\colluders = 1$ to $0.582$ at $\colluders = 20$. In binary-equivalent advantage units, $2(\mathrm{AUC}-1/2)$, this grows from $0.016$ to $0.164$. \textbf{The latter is a reference/audit-monitoring signal, not an attacker capability.} Production multi-tenant RAG without an audit primitive is harder to attack \emph{via top-$K$ alone} than the score-release class would imply; the audit's role is to certify the noise-then-select internals whose DP cost is measured by the score channel, not to disclose those scores.

\begin{figure}[t]
    \centering
    \includegraphics[width=\columnwidth]{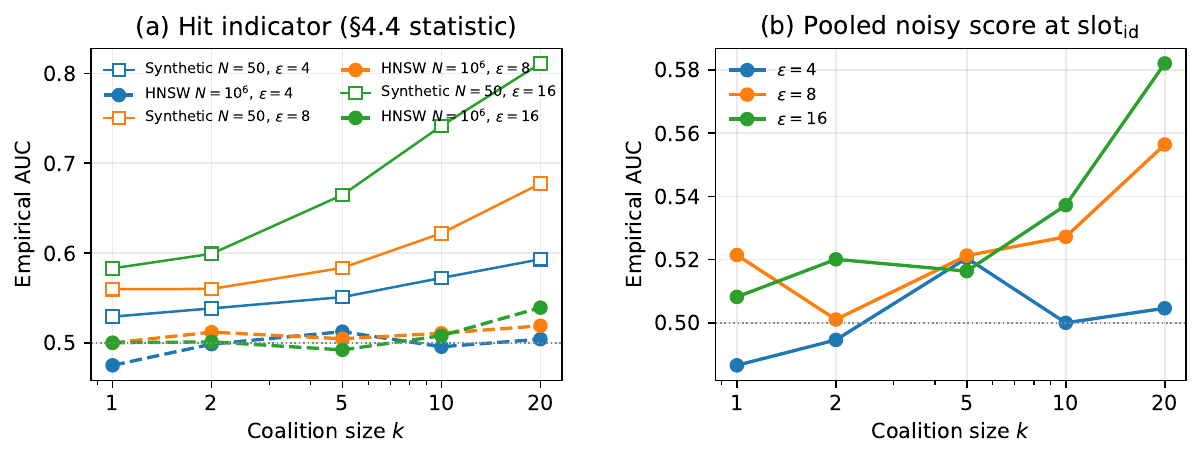}
    \caption{Production-scale HNSW on $10^6$ MS MARCO passages, \texttt{bge-small-en-v1.5} embedder, HNSW $(M{=}64,\mathsf{ef}_{\text{cstr}}{=}200,\mathsf{ef}_{\text{q}}{=}128)$, $\Delta = 0.133$, $T = 2{,}000$. (a)~\emph{User-observable} hit-indicator AUC: chance at scale (curves flat across $\colluders$). (b)~\emph{Instrumented} pooled noisy-score at the planted slot: monotone curve preserved at $\eps_{\text{acc}} \in \{8, 16\}$. The two panels show different observation surfaces: (a) is what an attacker account sees; (b) is an offline reference view of the attested noised-score channel, not an audit output that the verifier must see.}
    \label{fig:hnsw-prod}
\end{figure}

\subsection{What the experiment does not test}
\label{subsec:exp-limits}

Two caveats bound this empirical claim. \emph{First}, the alternative adversaries of \Cref{subsec:adversarial-robustness} cover statistic-side (Bayes LR) and probe-side (multi-target split) variants but not adaptive query loops; \Cref{thm:upper} bounds any such adaptive adversary unconditionally. \emph{Second}, the multi-target sweeps in \Cref{subsec:real-embedder,subsec:hnsw-production} cover one $\Delta$ window each ($[0.41, 0.49]$ on the toy corpus; ${\approx}[0.10, 0.14]$ on MS MARCO); a full cross-corpus, continuous-$\Delta$ embedder-distribution sweep is outside present scope.

\section{Audit Protocol}
\label{sec:audit-protocol}

This section specifies a three-party audit protocol that allows a verifier $\Auditor$ to certify the joint leakage bound of \Cref{cor:tightness} against a deployed RAG service without requiring the provider to disclose the index. The protocol is the contribution that distinguishes our work from purely defensive RAG-DP literature (\textit{e.g.,} \parencite{Cheng2025-remoterag}): existing defenses modify the retrieval algorithm to leak less per query; ours leaves the retrieval decision rule unchanged and instead wraps it in provider-side audit instrumentation (commitments, per-query ZK attestations, a coalition-size estimator over the ledger contents) that produces a verifiable leakage bound that any third party can check.

\subsection{Toolkit boundary}
\label{subsec:toolkit-boundary}

The toolkit composes five generic ZK-cryptographic primitives factored into upstream \texttt{cryptographic-audit-protocols} --- G1: ZK commitment over hashable artifacts; G2: append-only Merkle ledger; G3: ZK ``$f$ with public $\theta$ applied to private $x$'' template; G4: committed-seed Gaussian noise with hidden seed (\Cref{rem:hidden-seed}); G5: prover/verifier transcript schema --- with six RAG-specific primitives shipping in this paper: R1: embedder commitment $C_{\text{emb}}$; R2: query ledger + query-time receipts (\ref{p:receipts}); R3: noise-then-select attestation $\pi_{\text{ord}}$; R4: cross-tenant containment $\pi_{\text{tenant}}$; R5: coalition-size estimator $\hat{\colluders}$ with certificate $\pi_{\text{coal}}$; R6: per-tenant index-content vector commitment $C_{\text{idx}}$. Each R-primitive composes one or more G-primitives.

Three parties: \emph{Provider} $\Provider$ (knows the index $\Index$, embedder weights, and secret randomness used for DP noise; wishes to demonstrate compliance without revealing $\Index$); \emph{Verifier} $\Auditor$ (third party --- regulator, auditor, customer's representative --- with no prior trust relationship to $\Provider$; outputs a binary verdict and a leakage bound $\eps_{\text{audit}}$); \emph{Observer} $\mathcal{O}$ (anyone who reads the public ledger; may submit audit queries to test the protocol).

The protocol is parameterized by a hash function $H$ (instantiated as SHA-256, NIST FIPS 180-4, collision-resistant), a non-interactive zero-knowledge proof system $\Pi$ (e.g., Groth16 or PLONK; we treat it abstractly), the public DP parameters $(\eps_{\text{acc}}, \delta_{\text{acc}}, \sigma)$ from \Cref{sec:threat-model}, a policy-declared coalition cap $\colluders_{\max}$, and a policy-declared joint failure-probability target $\delta_{\text{policy}} \geq \delta_{\text{acc}}$ at which the auditor reports the joint $(\eps_{\text{audit}}, \delta_{\text{policy}})$-collusion-DP guarantee. Both $\colluders_{\max}$ and $\delta_{\text{policy}}$ are fixed at the audit-window opening as part of $C_{\text{policy}}$ (\Cref{subsec:audit-commit}).

\subsection{Phase~A: commitment}
\label{subsec:audit-commit}

Before the audit window $W$ opens, $\Provider$ publishes:
\begin{align}
C_{\text{emb}} &= H\bigl(\mathsf{Emb} \,\Vert\, \mathsf{prompt\_tmpl} \,\Vert\, \mathsf{tokenizer} \,\Vert\, \text{schema}(\Index)\bigr) \tag{R1} \\
C_{\text{idx}} &= \{C_{\Index_t}\}_{t \in \mathcal{T}}, \quad C_{\Index_t} = \mathsf{VCom}\bigl(\{\mathsf{Emb}(d) : d \in \Index_t\}\bigr) \tag{R6} \\
C_{\text{seed}} &= H(s_W),\quad s_W \sim \mathcal{U}(\{0,1\}^\lambda),\; \lambda \geq 256 \tag{G4} \\
C_{\text{ledger}, 0} &= H(\bot) \tag{G2} \\
C_{\text{policy}} &= H\bigl(\eps_{\text{acc}} \,\Vert\, \delta_{\text{acc}} \,\Vert\, \sigma \,\Vert\, \colluders_{\max} \,\Vert\, \delta_{\text{policy}} \,\Vert\, W\bigr)
\end{align}
$C_{\text{emb}}$ binds embedder, prompt template, tokenizer, and per-tenant schema (cardinalities, boundaries) --- publishable metadata. $C_{\text{idx}}$ binds per-tenant document embeddings: verifier learns only $|\Index_t|$, but every per-query proof must open positions of $C_{\Index_{t(a)}}$ rather than a free witness index. $\mathsf{VCom}$ is instantiable as a Merkle commitment over $H(\mathsf{Emb}(d))$ with succinct $\mathcal{O}(\log |\Index_t|)$ openings. $C_{\text{seed}}$ binds a fresh $\lambda$-bit seed $s_W$ driving per-query noise (G4); pre-window publication prevents post-hoc noise selection. $C_{\text{policy}}$ binds privacy parameters + audit window against retroactive renegotiation, with the well-formedness constraint $\delta_{\text{policy}} > \colluders_{\max} \delta_{\text{acc}}$ (required for the audit's $\delta$ reparametrization in \Cref{subsec:audit-verify}).

\subsection{Phase~B: per-query attestation}
\label{subsec:audit-attest}

For each query $q$ submitted by account $a$ during $W$, $\Provider$ runs the standard retrieval pipeline (\Cref{sec:threat-model}) and produces a zero-knowledge proof $\pi_q$ attesting four claims jointly:

\begin{enumerate}[leftmargin=*,label=(A\arabic*)]
\item \label{att:emb} \emph{Embedder consistency}: the embedding of $q$ was computed under the artifact bound by $C_{\text{emb}}$. (G3 instantiated with $f = \mathsf{Emb}$, public input $C_{\text{emb}}$, private witness $\mathsf{Emb}$.)

\item \label{att:noise} \emph{Noise derivation from committed seed} (G4): the proof attests that $z_q = \mathsf{DGS}(H(s_W \,\Vert\, \mathsf{rec}_q); \sigma)$ for $\sigma$ in $C_{\text{policy}}$, with $\mathsf{rec}_q = (a, t(a), H(q), \mathsf{ts})$ public and $s_W$ a private witness consistent with $C_{\text{seed}} = H(s_W)$. $\mathsf{DGS}$ is a deterministic discrete Gaussian sampler whose encoding is part of the audit-protocol specification. The seed $s_W$ is not opened to $\Auditor$; pre-window binding of $C_{\text{seed}}$ plus ZK soundness of the relation pin $z_q$ to the sampler output on the committed input, while keeping the noise's random source hidden from $\Auditor$'s view --- a necessary condition for the underlying Gaussian mechanism to retain its $(\eps_q, \delta_q)$-DP guarantee (\Cref{rem:hidden-seed}).

\item \label{att:order} \emph{Noise-then-select over committed index} (R3 + R6): top-$K$ indices are the $K$ largest of $\mathsf{Score}(q, E_{t(a)}) + z$, where $E_{t(a)} = \{\mathsf{Emb}(d) : d \in \Index_{t(a)}\}$ is the per-tenant embedding vector opened from $C_{\Index_{t(a)}}$. The proof carries openings sufficient to attest top-$K$, binding retrieval to the deployed index, tying the deployment to the per-query DP guarantee \eqref{eq:per-query-budget}.

\item \label{att:tenant} \emph{Cross-tenant containment} (R4): every position opened in $\pi_q$ comes from $C_{\Index_{t(a)}}$, never from $C_{\Index_{t'}}$ with $t' \neq t(a)$, i.e., the response respects the tenant boundary of the issuing account.
\end{enumerate}

\begin{remark}[Hidden-seed design and DP correctness]
\label{rem:hidden-seed}
$(\eps, \delta)$-DP of the Gaussian mechanism is a statement about the joint view of the auditor and the mechanism's transcript; if the auditor learns the random coins $s_W$ that produced $z_q$, then $\bar y_q = \mathsf{Score}(q, x) + z_q$ together with $z_q$ trivially reveals the unnoised $\mathsf{Score}(q, x)$, collapsing the DP guarantee to vacuous. The natural workaround (`\(\Provider\) opens $s_W$ at audit close so the verifier can re-derive $z_q$') therefore breaks the DP semantics that the audit is supposed to certify. We resolve this by keeping $s_W$ as a private witness inside $\pi_{\text{noise}}^{(q)}$: $\Auditor$ accepts the per-query noise relation under ZK soundness of $\Pi$ without ever seeing $s_W$ or $z_q$. The verifier's view is then a hash of $s_W$ plus a sequence of proof transcripts that are zero-knowledge in $s_W$; the DP mechanism's hidden-randomness condition is preserved by the simulator guarantee of $\Pi$. This design choice does not increase per-query circuit size relative to the cost model of \Cref{subsec:audit-cost} --- the noise-derivation in-circuit check (G4) is already required to bind $z_q$ to the committed seed; we are removing the redundant audit-close opening, not adding a new constraint.
\end{remark}

$\Provider$ appends $r_q = (a,\, t(a),\, H(q),\, H(\text{top-}K),\, \mathsf{ts},\, \pi_q)$ to the Merkle ledger (G2), publishes the new root $C_{\text{ledger}, t}$, and co-issues to $a$ a signed query-time receipt\label{p:receipts}
\(\rho_q = (a,\, H(q),\, H(\text{top-}K),\, \mathsf{pos}_q,\, C_{\text{ledger}, t},\, \mathsf{Sig}_{\Provider})\)
where $\mathsf{pos}_q$ is the leaf position and $\mathsf{Sig}$ uses a long-term identity key. Honest accounts (M5) submit receipts to $\Auditor$ at audit close; a missing or contradicted receipt is a non-inclusion witness against the ledger. Past records are provable in $\mathcal{O}(\log n)$ time; raw content is stored only as hashes.

\subsection{Phase~C: coalition-size estimation}
\label{subsec:audit-coalition}

At the end of the audit window, $\Provider$ commits a coalition-size estimate $\hat{\colluders}$ together with a certificate $\pi_{\text{coal}}$ (R5) attesting that an agreed-upon estimator was applied to the ledger contents.

\paragraph{Estimator (abuse-detection signal)} We use a query-correlation detector: cluster queries in $W$ by hash collisions and approximate-nearest-neighbor proximity in $\mathsf{Emb}$'s representation (which $C_{\text{emb}}$ binds, so the verifier can re-derive the clustering). Clusters of $\geq 2$ distinct accounts within a sliding window of length $\tau$ are candidate coalitions; let $c$ be the largest such cluster. The audit's pass condition is $c \leq \colluders_{\max}$. The bound $\eps_{\text{audit}}$ in \eqref{eq:eps-audit} is computed from $\colluders_{\max}$, not from $c$: the estimator gates policy compliance but is not a soundness condition on the DP guarantee. \Cref{subsec:audit-security} discusses adversarial-evasion limits.

\paragraph{Calibration} We calibrate the cosine-similarity threshold $\theta$ on a toy harness with $A = 30$ accounts, $n = 100$ queries per account, embedding dimension $d = 32$, sweeping $\theta \in [0.20, 0.95]$ over 16 grid points with $200$ trials per cell. The null distribution draws iid uniform unit-sphere queries; the alternative replaces $k_{\text{true}} \in \{2,5,10,20\}$ accounts with a coalition that follows one of three patterns: (\textbf{P-A}) all colluders issue an identical probe vector $q^*$; (\textbf{P-B}) jittered probes $q_t^{(i)} = \mathsf{normalize}(q^* + \zeta \eta_t^{(i)})$ with $\zeta = 0.10$ and $\eta \sim \mathcal{N}(0, I)$; (\textbf{P-C}) probes drawn iid from a fixed paraphrase intent set of size $5$. We declare coalition detection when $\hat{\colluders} \geq 2$ and pick the smallest $\theta$ achieving null FPR $\leq 5\%$.

\Cref{fig:estimator-calibration}(a) shows the null FPR as a function of $\theta$: below $\theta \approx 0.70$ the random unit-sphere null is over-connected (FPR $= 1$), the transition is sharp ($0.595$ at $\theta = 0.75$), and the operating point lands at $\theta^* = 0.80$ with achieved FPR $= 0.040$. \Cref{fig:estimator-calibration}(b) shows $\hat{\colluders}$ recovery against $\theta$ at $k_{\text{true}} = 10$: the estimator transitions from a single merged cluster ($\hat{\colluders} = A$) to exact recovery ($\hat{\colluders} = 10$) over the same window. We lift the resulting operating-point performance into a formal statement:

\begin{observation}[Estimator calibration on P-A/B/C]
\label{lem:estimator-guarantee}
At the calibrated threshold $\theta^* = 0.80$, the coalition-size estimator of \Cref{subsec:audit-coalition} satisfies the following on the toy harness ($A = 30$, $n = 100$, $d = 32$):
\begin{enumerate}[leftmargin=*,label=(\roman*)]
    \item \emph{False-positive rate on null queries:} FPR $\leq 0.05$ ($200$ trials of iid uniform unit-sphere queries; achieved FPR $0.040$);
    \item \emph{True-positive rate on coalition patterns P-A, P-B, P-C:} TPR $= 1.00$ at $k_{\text{true}} \in \{2, 5, 10, 20\}$ ($200$ trials per $(k_{\text{true}}, \text{pattern})$ cell; $\hat{\colluders}$ recovers $k_{\text{true}}$ exactly with zero variance across all $12$ cells).
\end{enumerate}
\end{observation}

\Cref{lem:estimator-guarantee} characterizes the estimator's behavior on the three pattern classes we calibrated against; the limitation is that all three share high within-coalition cosine similarity. An adversary who deliberately diversifies queries below $\theta$ in the embedder's representation evades detection by construction --- this is the residual gap that \Cref{subsec:audit-security} acknowledges and that motivates the embedder-binding step of \Cref{subsec:audit-commit}. Crucially, \Cref{lem:estimator-guarantee} is an \emph{operational} guarantee on detection performance, not a soundness condition on the DP bound \eqref{eq:eps-audit}: \Cref{thm:audit-soundness} below holds for any true coalition size $c_{\text{true}} \leq \colluders_{\max}$ regardless of pattern class. The estimator's role is to gate policy compliance; when it fires, the audit fails and the deployment is flagged.

\begin{figure}[t]
    \centering
    \includegraphics[width=\columnwidth]{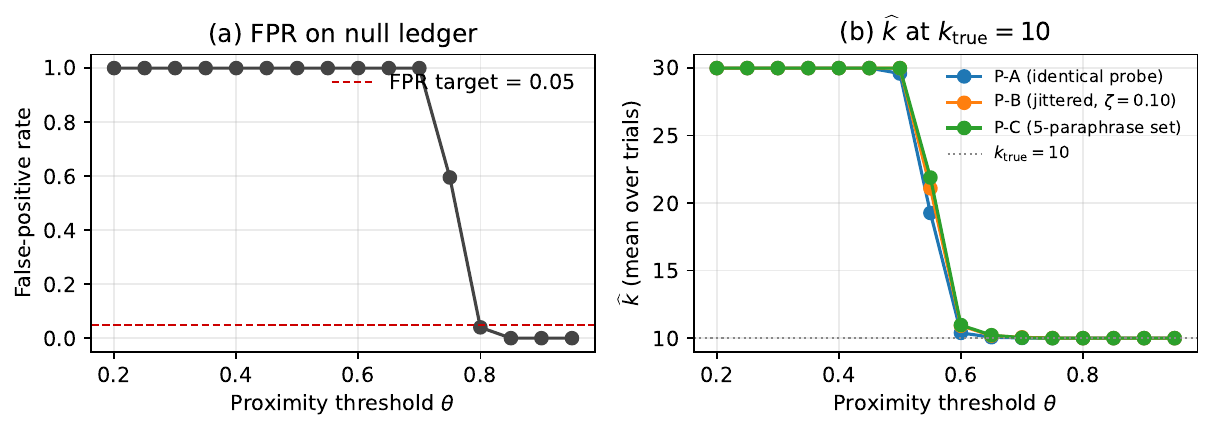}
    \caption{Coalition-size estimator calibration on the toy harness ($A = 30$, $n = 100$, $d = 32$, $200$ trials per cell).
    \textbf{(a)} Null FPR as a function of the proximity threshold $\theta$; the operating point $\theta^* = 0.80$ achieves FPR $= 0.040$.
    \textbf{(b)} Mean $\hat{\colluders}$ at $k_{\text{true}} = 10$ for the three attack patterns; the estimator transitions from a merged graph ($\hat{\colluders} = A$ at low $\theta$) to exact recovery ($\hat{\colluders} = 10$) across the same window. At $\theta^*$, TPR $= 1.00$ at all four tested $k_{\text{true}} \in \{2,5,10,20\}$ for all three patterns, with $\hat{\colluders}$ exactly recovering $k_{\text{true}}$.}
    \label{fig:estimator-calibration}
\end{figure}

\paragraph{Certificate} $\pi_{\text{coal}}$ is a ZK proof that $\hat{\colluders}$ is the output of the agreed estimator on the ledger contents (witness: the queries and their cluster assignments; public input: $C_{\text{ledger}, W}$ and the estimator's parameters bound by $C_{\text{policy}}$). This forces $\Provider$ to commit to a coalition estimate consistent with the public ledger.

\subsection{Phase~D: verification}
\label{subsec:audit-verify}

The verifier $\Auditor$ executes:
\begin{enumerate}[leftmargin=*]
\item \emph{Commitment integrity}: all of $C_{\text{emb}}, C_{\text{idx}}, C_{\text{seed}}, C_{\text{policy}}, C_{\text{ledger}, 0}$ were published before $W$ opened. For each sampled proof $\pi_q$, $\Auditor$ verifies the per-query noise-correctness component $\pi_{\text{noise}}^{(q)}$ (G4) against $C_{\text{seed}}$ and the public record $\mathsf{rec}_q$, without learning $s_W$. Index openings in $\pi_q$ verify against $C_{\Index_{t(a)}}$.
\item \emph{Per-query soundness}: for each $r_q$ in the ledger, $\Pi.\mathsf{Verify}(\pi_q, \cdot) = 1$ for claims \ref{att:emb}--\ref{att:tenant}. Sampled mode: $\Auditor$ verifies $s$ uniformly-random records and reports a property-testing guarantee --- any $\beta$-fraction of non-conforming records is detected with probability $\geq 1 - (1-\beta)^s$ (\Cref{subsec:audit-cost}). This does not certify worst-case soundness; rare bad events ($\beta$ near $1/|W|$) need full verification.
\item \emph{Ledger integrity + receipt consistency}: $C_{\text{ledger}, W}$ is a valid append-only extension of $C_{\text{ledger}, 0}$ (G2). For each submitted receipt $\rho_q$, $\Auditor$ verifies $\mathsf{Sig}_{\Provider}$ + Merkle inclusion at $\mathsf{pos}_q$ in $C_{\text{ledger}, W}$, and that the receipted $(H(q), H(\text{top-}K))$ match the ledger record. A receipt without a consistent ledger record is a forgery witness; the audit fails.
\item \emph{Policy-cap check}: $\Pi.\mathsf{Verify}(\pi_{\text{coal}}, \cdot) = 1$ and the estimator's largest valid cluster $c \leq \colluders_{\max}$. If $c > \colluders_{\max}$ the audit fails: the deployment violates its declared coalition cap.
\item \emph{Leakage bound}: $\delta_{\text{policy}}$ is the target final joint failure probability. Set the composition slack $\delta_{\text{comp}} = \delta_{\text{policy}} - \colluders_{\max} \delta_{\text{acc}}$ ($C_{\text{policy}}$ enforces $\delta_{\text{policy}} > \colluders_{\max} \delta_{\text{acc}}$). Plug $\colluders_{\max}, \delta_{\text{comp}}$ into \Cref{thm:upper}:
\begin{equation}
\label{eq:eps-audit}
\eps_{\text{audit}} \;=\; \sqrt{\colluders_{\max}} \cdot \eps_{\text{acc}} \cdot \sqrt{\frac{\log(1/\delta_{\text{comp}})}{\log(1/\delta_{\text{acc}})}} + \mathcal{O}\!\left(\frac{\colluders_{\max} \eps_{\text{acc}}^2}{\log(1/\delta_{\text{acc}})}\right).
\end{equation}
The reparametrization makes the final joint failure probability exactly $\delta_{\text{policy}}$; $\colluders_{\max}$ rather than $c$ keeps soundness independent of estimator adversarial robustness.
\end{enumerate}

$\Auditor$ outputs $(\textsf{PASS}, \eps_{\text{audit}})$ if all checks pass, else $(\textsf{FAIL}, \cdot)$.

\subsection{Security argument}
\label{subsec:audit-security}

\begin{theorem}[Audit soundness --- formal DP guarantee]
\label{thm:audit-soundness}
\emph{Meaningful-verdict conditions:} non-trivial $\colluders_{\max}$, M4 (\Cref{rem:access-control}) for cross-tenant or external regimes; generation-channel privacy is out of scope (\Cref{subsec:scope-of-dp}). Under these and $\Pi$ sound and zero-knowledge (the ZK property is required to keep $s_W$ hidden in $\pi_{\text{noise}}$; see \Cref{rem:hidden-seed}), $H$ collision-resistant, $\mathsf{VCom}$ binding, $\mathsf{DGS}$ statistically close to $\mathcal{N}(0,\sigma^2 I)$: \textbf{Mode~A.} With full per-record verification and at least one honest-account receipt sample (M5): if $\Auditor$ outputs $(\textsf{PASS}, \eps_{\text{audit}})$ \emph{and the true coalition size $c_{\text{true}} \leq \colluders_{\max}$}, the deployed retrieval-channel mechanism $\mathsf{Retr}$ \eqref{eq:retrieval} is $(\eps_{\text{audit}}, \delta_{\text{policy}})$-collusion-DP against any single tenant in $W$. \textbf{Mode~B.} With $s$ uniformly-sampled records: any $\beta$-fraction violation of \ref{att:emb}--\ref{att:tenant} is detected with probability $\geq 1 - (1-\beta)^s$; $\eps_{\text{audit}}$ is conditional on conformance of unsampled records. The bound is independent of the empirical estimator output $\hat{\colluders}$ and of the coalition's pattern class.
\end{theorem}

\begin{proof}[Sketch -- full proof in appendix]
Soundness of $\Pi$ + \ref{att:emb} forces the artifact bound by $C_{\text{emb}}$. \ref{att:noise} + pre-window $C_{\text{seed}}$ pins $z_q$ to the sampler output on the committed seed and ledger record, statistically close to a fresh $\mathcal{N}(0,\sigma^2 I)$ draw; ZK of $\Pi$ ensures $\Auditor$'s view is independent of $s_W$ given $C_{\text{seed}}$, preserving the DP mechanism's hidden-randomness assumption (\Cref{rem:hidden-seed}). Claim \ref{att:order} pins each attested query to the Gaussian-noised score-vector retrieval mechanism of \Cref{def:gaussian-class}: scores are computed over the committed tenant index, independent Gaussian noise with parameter $\sigma$ from $C_{\text{policy}}$ is added coordinatewise, and top-$K$ selection is deterministic post-processing. The scalar mechanism \eqref{eq:lb-mechanism} is used only for the lower-bound construction, not as the audited production mechanism; by \Cref{def:gaussian-class} and \eqref{eq:per-query-budget}, each attested query is $(\eps_q, \delta_q)$-DP. Claim \ref{att:tenant} confines responses to $\Index_{t(a)}$, so neighboring-index changes outside tenant $t^*$ do not affect the transcript --- privacy reduces to the $\colluders$-collusion analysis of \Cref{sec:dp-analysis}. The bound $\eps_{\text{audit}}$ \eqref{eq:eps-audit} is computed from $\colluders_{\max}$ and the composition slack $\delta_{\text{comp}}=\delta_{\text{policy}}-\colluders_{\max}\delta_{\text{acc}}$; under the hypothesis $c_{\text{true}} \leq \colluders_{\max}$, \Cref{thm:upper} at $\colluders = \colluders_{\max}$ yields the claimed final $(\eps_{\text{audit}},\delta_{\text{policy}})$ bound. The policy-cap check rejects audits where the estimator's empirical cluster $c$ exceeds $\colluders_{\max}$, but soundness does not depend on $c$ tracking $c_{\text{true}}$.
\end{proof}

\begin{remark}[Operational audit claim]
\label{rem:operational-audit}
\Cref{thm:audit-soundness} establishes the DP guarantee \emph{conditional on} $c_{\text{true}} \leq \colluders_{\max}$. Enforcing this hypothesis in practice is the estimator's job: \Cref{lem:estimator-guarantee} shows that \emph{only on the synthetic patterns P-A, P-B, P-C}, at the calibrated threshold $\theta^* = 0.80$, the estimator catches every violation in our experimental scope (TPR $= 1.00$ over $200 \cdot 4 \cdot 3 = 2400$ trials). This is an empirical guarantee on a finite pattern class, not a soundness claim against arbitrary coalitions. For coalitions whose query distribution sits within these three patterns, a passing audit therefore implies $c_{\text{true}} \leq \colluders_{\max}$ with high empirical confidence; combining with \Cref{thm:audit-soundness} yields the $(\eps_{\text{audit}}, \delta_{\text{policy}})$-collusion-DP claim on the deployment \emph{against in-pattern adversaries only}. For out-of-pattern (adversarially diversified) coalitions, the estimator may produce a false negative: \Cref{thm:audit-soundness} still holds if $c_{\text{true}} \leq \colluders_{\max}$, but the audit may PASS on a deployment with $c_{\text{true}} > \colluders_{\max}$, in which case the DP guarantee does not apply. We do \emph{not} claim a deployment-level guarantee against adversaries outside P-A/B/C; closing this gap requires either an estimator with provable adversarial robustness or a regulatory floor on $\colluders_{\max}$.
\end{remark}

\paragraph{What the protocol does \emph{not} guarantee} See \Cref{rem:operational-audit} (estimator evasion) and \Cref{subsec:limitations} (ledger completeness via M5; embedder determinism via M3).

\subsection{Cost}
\label{subsec:audit-cost}

Full per-query verification is too expensive for production-rate windows, so the protocol is meant to be run in sampled mode. The optimized Groth16 model in \Cref{app:cost-model} gives roughly $8.8$M constraints and $\sim9$ s proving time per query at $N=10^5$ documents, with constant verifier time ($\sim8$ ms/proof). Sampling $s=\lceil\log(1/\eta)/\beta\rceil$ records detects any $\beta$-fraction violation with probability $1-\eta$; the policy point $\beta\geq0.01,\eta=2^{-20}$ gives $s\approx1{,}400$ and $\sim11$ s verifier time per audit window, independent of $|W|$.

\section{Related Work}
\label{sec:related}

RAG membership-inference work establishes the attack surface but is almost entirely single-account: S$^2$MIA, difficulty-calibrated MIA, generation-artifact MIA, mask-based MIA, stealthy low-query attacks, and differential-calibration MIA all operate at $\colluders=1$ \parencite{Zeng2024-good,Li2025-generating,Wang2025-leaks,Feng2025-leak,Liu2024-mask,Naseh2025-riddle,Gao2025-dcmi}. Our results make the coalition size explicit: their $k=1$ attacks are the base case, while the relevant production benchmark is $\sqrt{k}\eps_{\text{acc}}$. Embedding inversion is orthogonal \parencite{Morris2024-textembeddings}; it composes with, rather than replaces, the retrieval-channel leakage analyzed here.

The $\Theta(\sqrt{k})$ composition rate is standard DP theory, including advanced/concurrent composition, RDP, and privacy odometers \parencite{DworkRoth2014,DworkRothblumVadhan2010,KairouzOhViswanath2017,Mironov2017,AbadiCCS2016,VadhanWang2021-concurrent,Rogers2016-odometers}. Federated-learning collusion targets training-time aggregation \parencite{Pasquini2022-eluding}; RemoteRAG-style defenses modify retrieval \parencite{Cheng2025-remoterag}. Our contribution is to align the known composition phenomenon with multi-tenant RAG's noise-then-select retrieval channel and to turn it into a ZK-auditable predicate using standard SNARK/Merkle machinery \parencite{Groth2016,BenSasson2014-zerocash,Setty2020-spartan}.

\section{Limitations and Scope}
\label{sec:discussion}
\label{subsec:limitations}

\paragraph{Retrieval channel only}\label{subsec:scope-of-dp}
\Cref{thm:upper,thm:audit-soundness} certify $\mathsf{Retr}$ in \eqref{eq:retrieval}, not the full output $\mathsf{LLM}(q,\mathsf{Retr}(q,a))$; generation re-accesses private documents after selection and needs a separate audit predicate.

\paragraph{Coalition cap and estimator}
The theorem is conditional on the true coalition size $c_{\text{true}}\leq\colluders_{\max}$. The query-correlation estimator gates policy compliance and is empirically calibrated for P-A/B/C, but adversarially diversified coalitions can evade it. A passing audit is therefore a deployment-level guarantee only when the declared $\colluders_{\max}$ is meaningful and the adversary class is within the estimator's scope.

\paragraph{Sampling and implementation cost}
Sampled verification detects bulk non-conformance but not rare bad records; full per-record verification is currently too costly at production rates. The Groth16 numbers are model estimates, not a deployed microbenchmark.

\paragraph{Embedder trust}
The protocol binds the embedder through $C_{\text{emb}}$ but does not prove that the embedder is resistant to inversion or other representation-level leakage \parencite{Morris2024-textembeddings}.

\paragraph{Disclosure}\label{subsec:disclosure}
The issue is structural rather than implementation-specific: per-account DP should not be advertised as the operational boundary unless the service also declares and audits a coalition cap $\colluders_{\max}$.

\appendix

\section{Additional experimental results}
\label{app:additional-experiments}

\subsection{Trained-embedder check}
\label{subsec:real-embedder}

The experiments in §\ref{subsec:headline}--§\ref{subsec:rate-transfer} use random unit-sphere embeddings at dimension $d=32$. To check that the $\sqrt{\colluders}$ rate transfers to a trained embedder (and quantify the constant-factor change), we re-ran the scalar-mechanism attack of \Cref{subsec:headline} with \texttt{BAAI/bge-small-en-v1.5} (384-dim, MTEB-standard sentence embedder) over a hardcoded $50$-paragraph factual corpus. We select a target/decoy pair with realized score gap closest to $\Delta = 0.40$ (vs.\ $\Delta = 1.0$ in §\ref{subsec:headline}, the degenerate maximum-signal case): the picker selects ``The Earth orbits the Sun \ldots'' as target and ``The Renaissance began in Italy \ldots'' as decoy, yielding $\Delta = 0.41$. Probe vector is the target embedding ($q^* = x_*$, matching the lower-bound construction). Other parameters: $\eps_{\text{acc}} = 4$, $\delta_{\text{acc}} = 10^{-6}$, $n = 200$ queries per account, $T = 2\,000$ trials per cell.

The empirical advantage grows from $0.032$ at $\colluders = 1$ to $0.150$ at $\colluders = 20$ --- a $4.69\times$ growth against predicted $\sqrt{20} \approx 4.47$; the $\sqrt{\colluders}$ scaling holds within finite-sample noise. The absolute advantage sits below the $\Delta = 1$ random-embedder baseline: at $\colluders = 20$, bge-small AUC is $0.575$ vs random's $0.638$. The advantage ratio ($0.150 / 0.275 = 0.54$) tracks the realized-$\Delta$ ratio ($0.41 / 1.00$) within ${\sim}30\%$, consistent with the closed-form linear-in-$\Delta$ prediction.

\paragraph{Target-distribution robustness} To check that the single-pair number is not idiosyncratic, we repeat over $10$ random $(\text{target}, \text{decoy})$ pairs from the same corpus with $\Delta \in [0.41, 0.49]$ at $\eps_{\text{acc}} = 4$, $T = 2\,000$ per cell.\footnote{Driver \texttt{run\_real\_embedder\_multi\_target.py}; full paths in the artifact repository.} The empirical AUC distribution across pairs is tight: at $\colluders = 1$ the median is $0.515$ (IQR $[0.504, 0.520]$, full range $[0.498, 0.524]$); at $\colluders = 20$ the median is $0.573$ (IQR $[0.563, 0.577]$, full range $[0.558, 0.581]$). The median advantage grows by a factor of $4.93\times$ across the sweep, slightly exceeding the predicted $\sqrt{20}\approx 4.47$ due to mild $\Phi$-saturation; the single-pair number ($0.575$ at $\colluders = 20$) sits inside the multi-pair range. The $\sqrt{\colluders}$ rate transfer is not an artifact of a single lucky pair.

\subsection{Alternative-adversary robustness check}
\label{subsec:adversarial-robustness}

To check whether any nearby alternative adversary beats the $\sqrt{\colluders}$ rate, we run two variants against the scalar mechanism at $\eps_{\text{acc}}=4$. (i)~\emph{Bayes LR}: replaces the pooled-mean statistic with the exact log-likelihood-ratio test; because LR is a monotone function of the pooled mean, AUC matches baseline to four decimal places at every $\colluders \in \{1, 2, 5, 10, 20\}$ --- the pooled-mean test is already at the LR-family optimum. (ii)~\emph{Diversified multi-target split}: a fraction $\rho$ of the $\colluders$ colluders point at the true target while the remaining $\colluders(1-\rho)$ probe random orthogonal directions, modeling hedging without target knowledge. At $\colluders = 20$, baseline AUC $= 0.638$, $\rho = 0.50$ gives $0.593$, $\rho = 0.25$ gives $0.580$; orthogonal probes contribute zero signal and the wasted budget compounds. Neither variant exceeds the baseline rate; combined with \Cref{thm:upper}'s upper bound, $\sqrt{\colluders}$ is empirically tight in the neighborhood of the §\ref{subsec:headline} attack. Adaptive query loops and embedder-inversion oracles are not exercised; \Cref{thm:upper} bounds both unconditionally.

\subsection{External vs.\ same-tenant collusion}
\label{subsec:external-collusion}

Under M4 (\Cref{rem:access-control}) the three collusion regimes reduce to the same per-query analysis at the same rate $\Theta(\sqrt{\colluders})$. We test the M4-failure regime empirically: external models $\colluders$ attacker accounts outside the victim tenant, each with a fresh per-account budget, querying the victim index via a shared retrieval backend (M4 active). Each attacker account applies fresh Gaussian noise per query calibrated by \eqref{eq:per-query-budget}; sweep parameters: $\colluders \in \{1,2,5,10,20\}$, $\eps_{\text{acc}} \in \{4, 8, 16\}$, $n = 200$, $T = 2{,}000$ trials, $50$ background docs, $K = 5$.

\Cref{fig:external-vs-same} shows the two regimes are statistically indistinguishable \emph{under M4}: across all 15 cells the absolute AUC difference is at most $0.022$ (median $0.010$), within one stderr ($\approx 0.011$). At the headline cell $\colluders=20, \eps_{\text{acc}}=16$, external AUC $= 0.810$ vs same-tenant $0.811$. Empirical $\sqrt{\colluders}$ growth across $\colluders=1{\to}20$: external $3.95\times$, same-tenant $4.11\times$, predicted $\sqrt{20}{\approx}4.47\times$. Read correctly: this quantifies external-collusion damage \emph{conditional on access-control failure}, not against ordinary external attackers, who have zero same-index access by \Cref{eq:retrieval} and zero leakage about $\Index_{t^*}$. Tenant-boundary access control \emph{is} the first line of defense; the audit-protocol coalition estimator (clustering by query similarity, not by tenant) is the second.

\begin{figure}[t]
    \centering
    \includegraphics[width=\columnwidth]{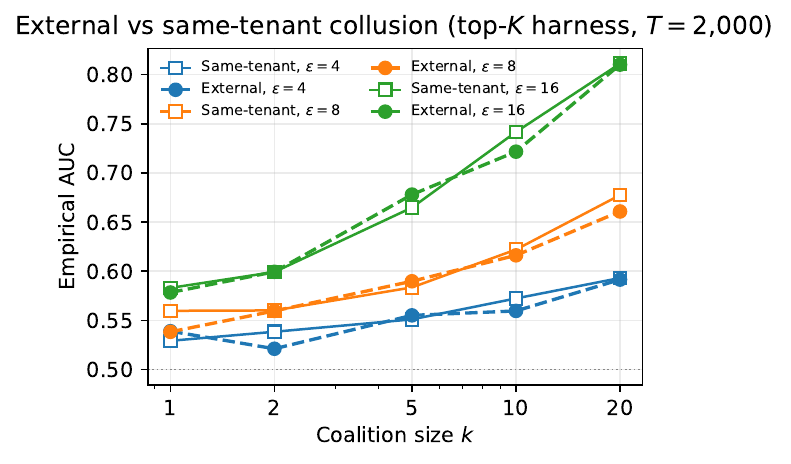}
    \caption{External vs same-tenant collusion against the top-$K$ harness at matched $(\colluders, \eps_{\text{acc}})$. Same-tenant baseline (open squares, solid) and external (filled circles, dashed) at three $\eps_{\text{acc}}$. Curves overlap across the full sweep; absolute AUC differences are bounded by $0.022$ (median $0.010$), within one stderr at $T = 2{,}000$.}
    \label{fig:external-vs-same}
\end{figure}

\section{Extended R\'enyi DP analysis}
\label{app:rdp-details}

\begin{definition}[R\'enyi DP \parencite{Mironov2017}]
\label{def:rdp}
For $\alpha > 1$, a mechanism $\mathcal{M}$ is $(\alpha, \rho(\alpha))$-RDP if for all neighboring $\Index \sim \Index'$,
\(D_\alpha(\mathcal{M}(\Index) \,\Vert\, \mathcal{M}(\Index')) \leq \rho(\alpha),\)
where $D_\alpha$ is the R\'enyi divergence of order $\alpha$.
\end{definition}

Three facts about RDP we use without proof (all from \parencite{Mironov2017}):
\begin{enumerate}[leftmargin=*,label=(R\arabic*)]
    \item \label{rdp:gauss-app} \emph{Gaussian mechanism.} The Gaussian mechanism with sensitivity $\Delta$ and noise $\mathcal{N}(0, \sigma^2)$ is $(\alpha, \alpha \Delta^2/(2\sigma^2))$-RDP for all $\alpha > 1$.
    \item \label{rdp:comp-app} \emph{Composition.} The $m$-fold adaptive composition of $(\alpha, \rho_i(\alpha))$-RDP mechanisms on the same database is $(\alpha, \sum_i \rho_i(\alpha))$-RDP, linear in $m$ with no $\sqrt{\log(1/\delta)}$ overhead.
    \item \label{rdp:conv-app} \emph{Conversion.} $(\alpha, \rho)$-RDP implies $(\rho + \log(1/\delta)/(\alpha-1), \delta)$-DP for any $\delta > 0$.
\end{enumerate}

\paragraph{Intermediate bound} The proof of \Cref{thm:upper-rdp} (\Cref{app:proof-upper-rdp}) yields the intermediate $\eps_\colluders^{\text{RDP}} = \Delta \sqrt{2\colluders n \log(1/\delta)}/\sigma + \colluders n \Delta^2 / (2 \sigma^2)$ before calibration substitution.

\paragraph{Improvement over \Cref{thm:upper}} The leading rate is identical, and the RDP route shrinks the residual by a constant factor of $4$ (from $\colluders \eps_{\text{acc}}^2/\log(1/\delta_{\text{acc}})$ to $\colluders \eps_{\text{acc}}^2/(4\log(1/\delta_{\text{acc}}))$, $n$-independent, linear in $\colluders$).

\paragraph{Numerical privacy tracking} The RDP form admits exact numerical evaluation via the moments accountant of \parencite{AbadiCCS2016}: we update $\rho(\alpha)$ on a grid $\alpha \in \{2, \ldots, 64\}$ per query and convert to $(\eps, \delta)$ at audit-window close. This provides tighter realized $(\eps, \delta)$ pairs than the closed-form bound at moderate $\colluders$ and is what the audit pipeline reports. We do not claim that query correlation reduces worst-case DP cost (it does not for a fixed Gaussian mechanism with fixed sensitivity); the tightening relative to the closed-form is purely arithmetic.

\section{Audit-protocol cost model}
\label{app:cost-model}

\begin{table}[h]
    \centering
    \caption{Estimated Groth16 per-query costs for the noise-then-select audit circuit at $d = 384$, top-$K = 5$. Optimized circuit pre-commits $C_{\text{sc},q}$, proves score-openings consistent with $C_{\Index_{t(a)}}$ on the $K$ selected positions, and additionally runs a Spartan-style score-correctness sumcheck attesting the committed score vector at \emph{all} $N$ positions; naive recomputes all $N$ inner products in-circuit. Proving and setup derived from $\sim 1\,\mu$s/constraint; verification constant per proof modulo $O(d \log N)$ sumcheck field ops.}
    \label{tab:zk-cost}
    \footnotesize
    \begin{tabular}{l r r r r}
    \toprule
    Mode & Index size $N$ & Constraints & Prove (s) & Setup (s) \\
    \midrule
    Optimized & $10^3$ & $9.8\!\times\!10^4$ & 0.1 & 0.05 \\
    Optimized & $10^4$ & $8.9\!\times\!10^5$ & 0.9 & 0.4 \\
    Optimized & $10^5$ & $8.8\!\times\!10^6$ & 8.8 & 4.4 \\
    Optimized & $10^6$ & $8.8\!\times\!10^7$ & 88 & 44 \\
    \midrule
    Naive & $10^3$ & $8.7\!\times\!10^5$ & 0.9 & 0.4 \\
    Naive & $10^4$ & $8.6\!\times\!10^6$ & 8.6 & 4.3 \\
    Naive & $10^5$ & $8.6\!\times\!10^7$ & 86 & 43 \\
    Naive & $10^6$ & $8.6\!\times\!10^8$ & 856 & 428 \\
    \bottomrule
    \end{tabular}
\end{table}

\section{Policy implications}
\label{app:policy}

The protocol's $(\textsf{PASS}, \eps_{\text{audit}})$ output fits a slot several frameworks contemplate but have not specified --- third-party verifiable guarantees without access to model or training data (EU AI Act Arts.~13/50, DSA Arts.~34--35) --- and we argue $\eps_{\text{audit}}$, not the per-account budget, is the right unit of disclosure. Absent mandate, three voluntary incentives apply: \emph{reputation} (verifiable trust for enterprise customers), \emph{early adoption} ahead of the AI Act's 2026--2027 high-risk phase-in, and \emph{liability reduction} from a contemporaneous defensible record. The protocol does not bind $\colluders_{\max}$: $\colluders_{\max} = 1$ trivially passes, recovering the status quo. Closing this gap requires regulatory floors or industry-consortium norms, outside our scope.

\section{Future work}
\label{app:future-work}

Beyond \Cref{subsec:limitations}: (i)~analytic-Gaussian RDP would close the residual constant-factor gap between \Cref{thm:upper-rdp} and \Cref{thm:lower}; (ii)~extension to RAG-augmented agents with persistent memory, where temporal accumulation composes with cross-account coalition; (iii)~Stackelberg analysis of the auditor-provider interaction is a natural sequel; (iv)~the most operationally important open question is whether the coalition-size estimator can be made adversarially robust without prohibitive ZK machinery.

\section{Full proofs}
\label{app:proofs}

This appendix expands the four proof sketches of \Cref{sec:dp-analysis,subsec:audit-security} into self-contained arguments. We retain the body sketches as quick references; nothing in the appendix changes the theorem statements.

\subsection{Proof of \Cref{lem:reduction} (Coordination reduction)}
\label{app:proof-reduction}

\begin{proof}
We construct an adaptive analyst $\Adv^\dagger$ that, given the same posterior and output channel as $\Adv$, produces a transcript with the same distribution. Under M1, if colluders shared randomness the simulator's $\colluders n$ independent draws would carry more information and the equality $\tau_\Adv \stackrel{d}{=} \tau_{\Adv^\dagger}$ would degrade to a one-sided inequality --- still a valid upper bound.

\textit{Construction.} $\Adv^\dagger$ schedules $\colluders$ virtual accounts $\{a_i\}$ with per-account rate $r$. At each round it chooses an account $a_i$ with nonzero rate budget and a query $q$ from $\Adv$'s posterior (\ref{cap:posterior}), issues $q$, receives response $r$, and appends $(a_i, q, r)$ to its transcript.

\textit{Distribution.} The Gaussian mechanism of \eqref{eq:lb-mechanism} has per-query privacy depending only on $q$ and $\Index_{t^*}$, not on the account label. So $\Pr[\tau_{\Adv^\dagger}]$ and $\Pr[\tau_\Adv]$ marginalize to the same distribution; capabilities \ref{cap:posterior}, \ref{cap:output}, \ref{cap:adaptive}, \ref{cap:rate-aware} make the conditioning, channel, and scheduling latitude identical. Joint transcripts agree.

\textit{Three collusion regimes.} Same-tenant collusion: queries target $\Index_{t^*}$ directly, the construction applies verbatim. Cross-tenant/external collusion: under M4 (\Cref{rem:access-control}), colluding accounts share retrieval access to $\Index_{t^*}$ and the reduction goes through identically; without M4 the regimes have zero leakage about $\Index_{t^*}$ and the reduction is vacuous.
\end{proof}

\subsection{Proof of \Cref{thm:upper}}
\label{app:proof-upper}

\begin{proof}
By \Cref{lem:reduction}, it suffices to upper-bound the privacy of $\colluders n$ adaptive $(\eps_q, \delta_q)$-DP queries on $\Index_{t^*}$.

\textit{Composition.} Apply \Cref{lem:advcomp} with $m = \colluders n$, $\eps_0 = \eps_q$, $\delta_0 = \delta_q$, $\delta'$ chosen below:
\begin{equation}
\eps_\colluders(\delta') \leq \sqrt{2\colluders n \log(1/\delta')} \cdot \eps_q + \colluders n \cdot \eps_q(e^{\eps_q} - 1).
\label{eq:app-advcomp-substituted}
\end{equation}
For $\eps_q \in (0, 1]$, $e^{\eps_q} - 1 \leq 2\eps_q$, so the second term is bounded by $2 \colluders n \eps_q^2$.

\textit{Substituting the per-query calibration.} Equation \eqref{eq:per-query-budget} gives $\eps_q = \eps_{\text{acc}} / \sqrt{2n\log(1/\delta_{\text{acc}})}$. Identify $\delta'$ with $\delta$ (the body theorem's parameter is the composition slack; the final joint failure probability is $\delta + \colluders \delta_{\text{acc}}$, tracked below). Then
\begin{align*}
\sqrt{2\colluders n \log(1/\delta)} \cdot \eps_q
&= \sqrt{2\colluders n \log(1/\delta)} \cdot \frac{\eps_{\text{acc}}}{\sqrt{2n\log(1/\delta_{\text{acc}})}} \\
&= \sqrt{\colluders} \cdot \eps_{\text{acc}} \cdot \sqrt{\frac{\log(1/\delta)}{\log(1/\delta_{\text{acc}})}},
\end{align*}
matching \eqref{eq:upper-bound}'s dominant term. The second term of \eqref{eq:app-advcomp-substituted} satisfies
\(
2 \colluders n \eps_q^2 = \frac{\colluders \eps_{\text{acc}}^2}{\log(1/\delta_{\text{acc}})}.
\)

\textit{Failure-probability accounting.} The total $\delta$ budget consists of (i) the composition slack $\delta$, and (ii) the per-query failure probabilities, $\colluders n \cdot \delta_q = \colluders \delta_{\text{acc}}$ by \eqref{eq:per-query-budget}. Summing,
\(
\delta_{\text{joint}} = \delta + \colluders \delta_{\text{acc}},
\)
which matches the theorem statement.
\end{proof}

\subsection{Proof of \Cref{thm:upper-rdp}}
\label{app:proof-upper-rdp}

\begin{proof}
By \Cref{lem:reduction}, the joint mechanism is the $\colluders n$-fold composition of the per-query Gaussian mechanism of \eqref{eq:lb-mechanism} on $\Index_{t^*}$.

\textit{RDP composition.} By the Gaussian-mechanism RDP fact~\parencite{Mironov2017}, the per-query mechanism is $(\alpha, \alpha \Delta^2 / (2\sigma^2))$-RDP for all $\alpha > 1$. By the RDP composition fact, the $\colluders n$-fold composition is $(\alpha, \colluders n \alpha \Delta^2 / (2\sigma^2))$-RDP.

\textit{Optimal-order conversion.} Apply the RDP-to-DP conversion~\parencite{Mironov2017} to $\eps_\colluders^{\text{RDP}}(\alpha, \delta) = \colluders n \alpha \Delta^2 / (2\sigma^2) + \log(1/\delta) / (\alpha - 1)$. Setting $\partial/\partial\alpha = 0$ gives $\alpha^* - 1 = \sigma\sqrt{2\log(1/\delta)/(\colluders n \Delta^2)}$; substituting yields
\(
\eps_\colluders^{\text{RDP}}(\alpha^*, \delta) = \Delta \sqrt{2 \colluders n \log(1/\delta)}/\sigma + \colluders n \Delta^2 / (2\sigma^2).
\)

\textit{Calibration substitution.} The same $\sigma$ from \eqref{eq:per-query-budget} ($\sigma = \Delta\sqrt{2n\log(1/\delta_{\text{acc}})}/\eps_{\text{acc}}$) satisfies the RDP route at $m=n$, equating to $(\eps_{\text{acc}}, \delta_{\text{acc}})$. Substituting in the bound above:
\(
\Delta\sqrt{2\colluders n \log(1/\delta)}/\sigma = \sqrt{\colluders} \, \eps_{\text{acc}} \sqrt{\log(1/\delta)/\log(1/\delta_{\text{acc}})}
\)
and the residual $\colluders n \Delta^2/(2\sigma^2) = \colluders \eps_{\text{acc}}^2 / (4 \log(1/\delta_{\text{acc}}))$, a constant-factor improvement (1/4 vs.\ 1) over \Cref{thm:upper}'s residual, with leading-order coefficient matching \eqref{eq:upper-bound} exactly. This gives \eqref{eq:rdp-final}.
\end{proof}

\subsection{Proof of \Cref{thm:audit-soundness}}
\label{app:proof-audit-soundness}

\begin{proof}
Five attestations compose. (A1)~Opening $C_{\text{emb}}$ at audit close, plus collision-resistance of $H$ and soundness of $\Pi$ on \ref{att:emb}, pins retrieval to the committed embedder. (A2)~\ref{att:noise} verifies $z_q = \mathsf{DGS}(H(s_W \,\Vert\, \mathsf{rec}_q); \sigma)$ in zero knowledge under $C_{\text{seed}}$; pre-window binding of $C_{\text{seed}}$ rules out post-hoc seed selection, so $z_q$ is statistically close to $\mathcal{N}(0, \sigma^2 I)$~\parencite{BenSasson2014-zerocash}. The zero-knowledge property of $\Pi$ keeps $s_W$ outside $\Auditor$'s view, preserving the Gaussian-mechanism DP guarantee on the noised score channel (\Cref{rem:hidden-seed}). (A3)~\ref{att:order} pins $\mathsf{TopK}$ to score-then-noise over the embeddings opened from $C_{\text{idx}}$; binding-of-$\mathsf{VCom}$ on $C_{\Index_{t(a)}}$ (R6) prevents the prover from substituting a private witness index for the deployed one. The optimized instantiation of \Cref{app:cost-model} adds an aggregated score-correctness sumcheck attesting the committed score vector $C_{\text{sc},q}$ against the deployed embedding at \emph{all} $N$ positions (Schwartz-Zippel soundness $1 - N/|\mathbb{F}|$), closing the score-substitution gap that affects partial-opening instantiations. (A4)~\ref{att:tenant} restricts every $C_{\text{idx}}$-opening to $C_{\Index_{t(a)}}$. (A5)~$\pi_{\text{coal}}$ attests $c \leq \colluders_{\max}$; \Cref{lem:estimator-guarantee} calibrates detection on the tested patterns P-A/B/C (FPR $\leq 0.05$, TPR $= 1$ for $k_{\text{true}} \in \{2,5,10,20\}$), with evasive coalitions a residual gap (\Cref{rem:operational-audit}, \Cref{subsec:audit-security}). Combining (A1)--(A4) pins every attested query to the Gaussian-noised score-vector retrieval class of \Cref{def:gaussian-class} on the deployed index: scores are computed against the committed tenant index, independent Gaussian noise with the policy-bound scale is added coordinatewise, and top-$K$ is deterministic post-processing. The scalar mechanism \eqref{eq:lb-mechanism} is only the lower-bound construction. Applying \Cref{thm:upper} at $\colluders=\colluders_{\max}$ with composition slack $\delta_{\text{comp}}=\delta_{\text{policy}}-\colluders_{\max}\delta_{\text{acc}}$ yields \eqref{eq:eps-audit} and final failure probability $\delta_{\text{policy}}$. The bound is independent of empirical $c$: the estimator gates, it does not enter soundness.
\end{proof}

\section{Ethics Considerations}
\label{sec:ethics}

We address the four categories of ethical concern relevant to IEEE S\&P submissions.

\paragraph{Human subjects.} This work involves no human-subjects research. All experiments use synthetic corpora and toy multi-tenant indices generated by the author; no user queries, user data, or any personally identifiable information is collected, stored, or analyzed. No IRB review was therefore required, and none was sought.

\paragraph{Live systems.} Empirical evaluation does not exercise any third-party deployed RAG service. The coordinated-query attacks of \Cref{sec:attacks} are run exclusively against a local harness (\texttt{src/rag\_collusion\_privacy\_audit/harness}) over synthetic indices. No production RAG provider is measured, probed, or exploited; the threat model in \Cref{sec:threat-model} is described at the architectural level and references public documentation only.

\paragraph{Vulnerability disclosure.} The gap is structural and derivable from public DP composition theory specialized to RAG (\Cref{subsec:disclosure}); targeted provider notification or a coordinated embargo of the kind appropriate for software bugs is therefore not warranted. The arXiv preprint is public alongside the paper; the audit-side toolkit (\texttt{cryptographic-audit-protocols}) accompanies the camera-ready release. This paper contains no CVE identifiers.

\paragraph{Dual-use considerations.} The coordinated-attack methodology and audit protocol described here are dual-use: the same construction that lets an auditor bound coalition leakage could in principle guide an adversary in selecting query distributions for membership inference. The author has weighed this against the benefits in two ways. First, the structural nature of the gap (\Cref{thm:lower}) means it is derivable from public DP composition theory; withholding the analysis does not meaningfully delay an attacker who reads the published DP literature. Second, the audit protocol itself (\Cref{sec:audit-protocol}) is asymmetric in favor of defenders: providers and auditors can deploy it unilaterally, while attackers gain little from knowing the protocol exists. The released toolkit (\texttt{cryptographic-audit-protocols}) ships only the audit-side primitives, not the attack harness. The net effect of publication is assessed as protective rather than enabling.

\paragraph{Generative AI usage.} Generative AI was used for editorial purposes in this manuscript, and all outputs were inspected by the author to ensure accuracy and originality. Generative AI was not used as part of the research methodology: theorem statements, proofs, identification strategy, experimental design, and code were authored and validated by the author without generative AI assistance.

\printbibliography

\end{document}